\documentclass[11pt,a4paper]{article}

\usepackage[utf8]{inputenc}

\usepackage{amssymb,amsfonts,amsmath}
\usepackage{amsthm} 

\usepackage{graphicx}

\usepackage{mathtools}      
\usepackage{upgreek}        
\usepackage{psfrag}         
\usepackage{multirow}       
\usepackage{xargs}          
\usepackage{color}          

\usepackage{url}
\usepackage{subfig}
\usepackage[colorlinks=true,citecolor=black,linkcolor=black,urlcolor=blue]{hyperref}
\usepackage{fullpage}
\usepackage{xcolor}
\usepackage[mathlines]{lineno}
\usepackage{authblk}

\newtheorem{theorem}{Theorem}[section]

\newtheorem{lemma}[theorem]{Lemma}

\newtheorem{definition}[theorem]{Definition}

\newtheorem{observation}[theorem]{Observation}

\providecommand{\prob}[1]{{\sf #1}}

\providecommand{\area}{\mathsf{Area}}
\providecommand{\weight}{\mathsf{Weight}}

\providecommand{\next}{\mathsf{next}}

\providecommand{\nil}{\mathbf{nil}}
\providecommand{\first}{\mathsf{first}}

\newcommand\blfootnote[1]{%
  \begingroup
  \renewcommand\thefootnote{}\footnote{#1}%
  \addtocounter{footnote}{-1}%
  \endgroup}

\title{Maximum rectilinear convex subsets}
\date{}

\author[1]{Hern\'an Gonz\'alez-Aguilar\thanks{Email: hernan@fc.uaslp.mx}}
\author[2]{David Orden \thanks{Research supported by project MTM2017-83750-P of the Spanish Ministry of Science (AEI/FEDER, UE). Email: david.orden@uah.es}}
\author[3]{Pablo P\'erez-Lantero\thanks{Research supported by project CONICYT FONDECYT/Regular 1160543 (Chile). Email: pablo.perez.l@usach.cl}}
\author[4]{David Rappaport\thanks{Research supported by NSERC of Canada Discovery Grant RGPIN/06662-2015. Email: daver@cs.queensu.ca}}
\author[5]{Carlos Seara\thanks{Research supported by projects MTM2015-63791-R MINECO/FEDER and Gen. Cat. DGR 2017SGR1640. Email: carlos.seara@upc.edu}}
\author[6]{Javier Tejel\thanks{Research supported by projects MTM2015-63791-R MINECO/FEDER and Gobierno de Arag\'on E41-17R. Email: jtejel@unizar.es}}
\author[7]{Jorge Urrutia\thanks{Research supported by PAPIIT grant IN102117 from UNAM. Email: urrutia@matem.unam.mx}}

\affil[1]{Facultad de Ciencias, Universidad Aut\'onoma de San Luis Potos\'i, Mexico.}
\affil[2]{Departamento de F\'{i}sica y Matem\'{a}ticas, Universidad de Alcal\'{a}, Spain.}
\affil[3]{Departamento de Matem\'atica y Ciencia de la Computaci\'on, USACH, Chile.}
\affil[4]{School of Computing, Queen's University, Canada.}
\affil[5]{Departament de Matem\`{a}tiques, Universitat Polit\`{e}cnica de Catalunya, Spain.}
\affil[6]{Departamento de M\'etodos Estad\'isticos, IUMA, Universidad de Zaragoza, Spain.}
\affil[7]{Instituto de Matem\'{a}ticas, Universidad Nacional Aut\'{o}noma de M\'{e}xico, Mexico.}

\begin{document}
  \maketitle
 \blfootnote{
 	\begin{minipage}[l]{0.3\textwidth}
 		\includegraphics[trim=10cm 6cm 10cm 5cm,clip,scale=0.12]{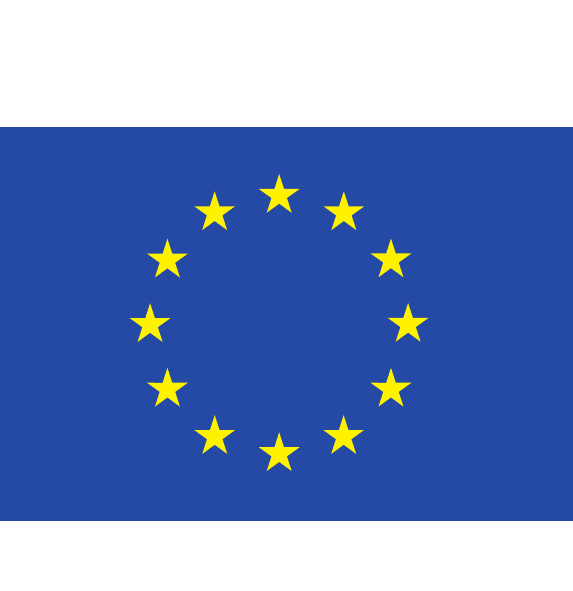}
 	\end{minipage}
 	\hspace{-3cm}
 	\begin{minipage}[l][1cm]{0.75\textwidth}
       	This work has received funding from the European Union's Horizon 2020
       	research and innovation programme under the Marie Sk\l{}odowska-Curie
       	grant agreement No 734922.
      \end{minipage}
 }
 \vspace{-0.9cm}

\begin{abstract}
Let $P$\label{page1} be a set of $n$ points in the plane. We consider a variation of the classical Erd\H{o}s-Szekeres problem, presenting efficient algorithms with $O(n^3)$ running time and $O(n^2)$ space complexity that compute: (1) A subset~$S$ of~$P$ such that the boundary of the rectilinear convex hull of $S$ has the maximum number of points from $P$, (2) a subset $S$ of $P$ such that the boundary of the rectilinear convex hull of $S$ has the maximum number of points from $P$ and its interior contains no element of $P$, (3) a subset $S$ of $P$ such that the rectilinear convex hull of $S$ has maximum area and its interior contains no element of $P$, and (4) when each point of $P$ is assigned a weight, positive or negative, a subset $S$ of $P$ that maximizes the total weight of the points in the rectilinear convex hull of $S$.

We also revisit the problems of computing a maximum-area orthoconvex polygon and computing a maximum-area staircase polygon, amidst a point set in a rectangular domain. We obtain new and simpler algorithms to solve both problems with the same complexity as in the state of the art.
\end{abstract}

%

\section{Introduction}

Let $P$ be a point set in general position in the plane. A subset $S$ of $P$ with $k$ elements is called a \emph{convex $k$-gon} if the elements of $S$ are the vertices of a convex polygon, and it is called a \emph{convex $k$-hole} of $P$ if the interior of the convex hull of $S$ contains no element of $P$. The study of convex $k$-gons and convex $k$-holes of point sets started in a seminal paper by Erd\H{o}s and Szekeres~\cite{ErdSze35} in 1935. Since then,  numerous papers about both the combinatorial and the algorithmic aspects of convex $k$-gons and convex $k$-holes have been published. The reader can consult the two survey papers about so-called Erd\H{o}s-Szekeres type problems~\cite{BraMosPac05,morris2000erdos}.

There are recent papers studying the existence and number of convex $k$-gons and convex $k$-holes for finite point sets in the plane~\cite{aichholzer2015k,aichholzer20144,aichholzer2014lower}. Papers dealing with the algorithmic complexity of finding largest convex $k$-gons and convex $k$-holes are, respectively, Chv\'atal and Kincsek~\cite{chvatal1980} and Avis and Rappaport~\cite{avis1985}, which solve these problems in $O(n^3)$ time.

Erd\H{o}s-Szekeres type problems have also been studied for colored point sets. Let $P$ be a point set such that each of its elements is assigned a color, say red or blue. Bautista-Santiago~et~al.~\cite{bautista2011} studied the problem of finding a monochromatic subset $S$ of $P$ of maximum cardinality such that all of the elements of $P$ contained in the convex hull of $S$ have the same color. As a generalization, they also studied the problem in which each element of $P$ has assigned a (positive or negative) weight. In this case, the goal is to find a subset $S$ of $P$ that maximizes the total weight of the points of $P$ contained in the convex hull of $S$.  Each of these problems was solved in $O(n^3)$ time and $O(n^2)$ space. In addition, their algorithm can  easily be adapted to find a subset $S$ of $P$ such that the convex hull of $S$ is empty and of maximum area in $O(n^3)$ time and $O(n^2)$ space.

In this paper, we study Erd\H{o}s-Szekeres type problems under a variation of convexity known as \emph{rectilinear convexity}, or \emph{orthoconvexity}: Let $P=\{p_1,\ldots,p_n\}$ be a set of $n$ points in the plane in general position. A \emph{quadrant} of the plane is the intersection of two open half-planes whose supporting lines are parallel to the $x$- and $y$-axes, respectively. We say that a quadrant $Q$ is $P$-\emph{free} if it does not contain any point of $P$. The \emph{rectilinear convex hull} of $P$, denoted as $RCH(P)$, initially defined by Ottmann~et~al.~\cite{ottmann1984}, is defined as:
\begin{displaymath}
	RCH(P) ~=~ \mathbb{R}^{2}-\bigcup_{\textrm{$Q$ is~$P$-\emph{free}}}Q.
\end{displaymath}
The rectilinear convex hull of a point set might be a simply connected set, yielding an intuitive and appealing structure (see Figure~\ref{fig:RCH1}). However, in other cases the rectilinear convex hull can have several connected components (see Figure~\ref{fig:RCH3}), some of which might be single points which we call {\em pinched} points. The \emph{size} of $RCH(P)$ is the number of elements of $P$ on the boundary of $RCH(P)$. The sizes of the rectilinear convex hulls in Figures~\ref{fig:RCH1} and~\ref{fig:RCH3} are, respectively, thirteen and twelve.

\begin{figure}[ht]
	\centering
	\subfloat[]{
		\includegraphics[scale=0.95,page=3]{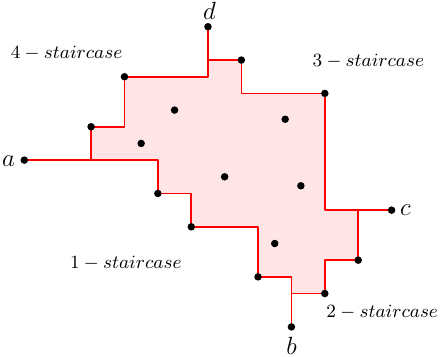}
		\label{fig:RCH1}
	}~~~~
	\subfloat[]{
		\includegraphics[scale=0.95,page=5]{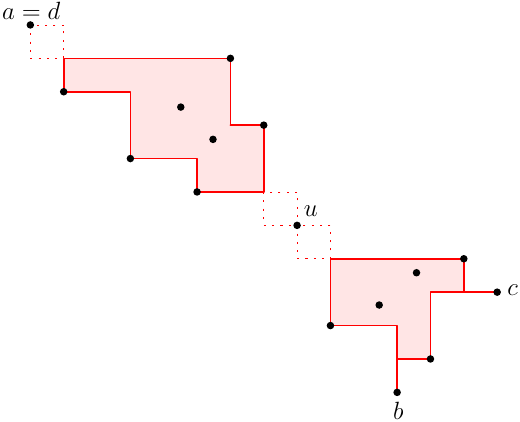}
		\label{fig:RCH3}
	}
	\caption{\small{
		(a) A point set with a connected rectilinear convex hull.
		(b) A point set whose rectilinear convex hull is disconnected,
		two of its components are pinched points.
	}}
	\label{fig:examples1}
\end{figure}

Alegr\'ia-Galicia et al.~\cite{alegria2017} gave an optimal $\Theta(n\log n)$-time and $O(n)$-space algorithm to compute the orientation of the coordinate axes such that the rectilinear convex hull of a set $P$ of $n$ points in the plane has minimum area. The reader can refer to the literature for other results related to rectilinear convexity~\cite{alegria2017,fink-wood2004,rawlins1988}. In this paper, we present efficient algorithms for the following geometric optimization problems:

\noindent\prob{MaxRCH}: Given a set $P$ of $n$ points in the plane, find a subset $S\subseteq P$ such that the size of $RCH(S)$ is maximized.

\noindent\prob{MaxEmptyRCH}: Given a set $P$ of $n$ points in the plane, find a subset $S\subseteq P$ such that the interior of $RCH(S)$ contains no point of $P$ and the size of $RCH(S)$ is maximized.

\noindent\prob{MaxAreaRCH}: Given a set $P$ of $n$ points in the plane, find a subset $S\subseteq P$ such that the interior of $RCH(S)$ contains no point of $P$ and the area of $RCH(S)$ is maximized.

\noindent\prob{MaxWeightRCH}: Given a set $P$ of $n$ points in the plane, such that each $p\in P$ is assigned a (positive or negative) weight $w(p)$, find a subset $S\subseteq P$ that maximizes $\sum_{p\in P\cap RCH(S)}w(p)$.

In Section~\ref{sec:MaxRCH}, we give an $O(n^3)$-time $O(n^2)$-space algorithm to solve the \prob{MaxRCH} problem. Then, in Section~\ref{sec:EmptyRCH} we show how to adapt this algorithm to solve the other three problems, each in $O(n^3)$ time and $O(n^2)$ space. The complexities of our algorithms are the same as the complexities of the best-known algorithms to solve these problems with the usual definition of convexity.

Besides presenting the first algorithms to solve the problems \prob{MaxRCH}, \prob{MaxEmptyRCH}, \prob{MaxAreaRCH}, and \prob{MaxWeightRCH}, we show that our techniques can be used to provide new algorithms to solve two additional problems, considered by Nandy~et~al.~\cite{nandy2010} and Nandy and Bhattacharya~\cite{nandy2003}.

\noindent\prob{MaxOrthoconvexPolygon}: Given a set $P$ of $n$ points in the plane, contained in an axis-aligned rectangle $\mathcal{R}$ called {\em domain}, find an orthoconvex polygon of maximum area that is contained in $\mathcal{R}$ and its interior contains no point of $P$,  where a polygon is {\em orthoconvex} if its sides are axis-parallel and its intersection with any horizontal or vertical line is empty, or a line segment~\cite{rawlins1988}.

\noindent\prob{MaxStaircasePolygon}: Given a set $P$ of $n$ points in the plane contained in an axis-aligned rectangle $\mathcal{R}$, find a staircase polygon of maximum area contained in $\mathcal{R}$ whose interior contains no point of $P$, where an orthoconvex polygon contained in $\mathcal{R}$ is called a {\em staircase polygon} if it contains two opposite corners of $\mathcal{R}$.

We show in Section~\ref{sec:ortho} how to solve the \prob{MaxOrthoconvexPolygon} problem in $O(n^3)$ time and $O(n^2)$ space, and we give in Section~\ref{sec:MESP} an $O(n^2)$ time and $O(n^2)$ space algorithm to solve the \prob{MaxStaircasePolygon} problem. The algorithms by Nandy~et~al.~\cite{nandy2010} and Nandy and Bhattacharya~\cite{nandy2003} have the same time and space complexities as ours, but we believe that our solutions are simpler and easier to understand. Orthoconvex polygons (and in general isothetic polygons) have different applications in areas as VLSI layout design or robotic visibility (see the references in~\cite{nandy2003} for more information regarding applications).

\section{Some notation and definitions}\label{sec:notation}

A summary of the main notation used can be found as an appendix. For the sake of simplicity, we assume that all point sets $P$ considered in this paper are in {\em general position}, which  means that no two points of $P$ share the same $x$- or $y$-coordinate. Using a $O(n\log n)$-time preprocessing step, we can also assume when necessary that the points of a point set $P$ are ordered by $x$-coordinate or $y$-coordinate. Given a point set $P$ in the plane, we will use $a$, $b$, $c$, and $d$\label{page2} to denote the leftmost, bottommost, rightmost, and topmost points of $P$, respectively, unless otherwise stated. Note that $a$, $b$, $c$, and $d$ are not necessarily different. In Figure~\ref{fig:RCH3}, we have $a=d$.

\begin{figure}[ht]
	\centering
	\subfloat[]{\includegraphics[scale=1.05,page=26]{img.pdf}\label{fig:1S}}~~~
	\subfloat[]{\includegraphics[scale=1.00,page=1]{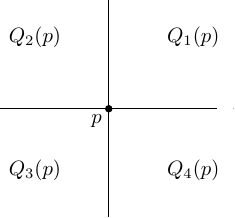}\label{fig:Qi}}~~~
	\subfloat[]{\includegraphics[scale=1.00,page=2]{img.pdf}\label{fig:Mi}}
	\caption{\small{(a) A $1$-staircase. (b) The definition of the sets $Q_i(p)$. (c) A 7-point set $P$ and the set $M_1(P)$. The vertices of the boundary of $M_1(P)$ in $P$ are the 1-extremal points of $P$. The thick polygonal line is the $1$-staircase associated with $P$.}}\label{fig:definition}
\end{figure}

Given a point $p$ of the plane, let $p_x$ and $p_y$\label{page3} denote the $x$- and $y$-coordinates of $p$, respectively.

\begin{definition}
For $p,q\in\mathbb{R}^2$, $p \neq q$, we write $p\prec q$\label{page4} to denote that $p_x < q_x$ and $p_y < q_y$, and $p\prec' q$\label{page5} to denote that $p_x < q_x$ and $p_y > q_y$.
\end{definition}

\begin{definition}
Let $p,q\in\mathbb{R}^2$, and consider a set $S=\{v_1, \ldots , v_k\}$ of $k$ points such that $v_1=p$, $v_k=q$, and $v_i \prec' v_{i+1}$ for $i=1,2,\dots,k-1$. A \emph{$1$-staircase} joining $p$ to $q$ is an orthogonal polygonal chain, such that two consecutive elements of $S$ are joined by an elbow consisting of a horizontal segment followed by a vertical segment. For an illustration, see Figure~\ref{fig:1S}.
\end{definition}

A \emph{$3$-staircase} joining $p$ to $q$ is defined in a similar way, but using elbows whose first segment is vertical. Analogously, we define $2$- and $4$-staircases, except that we require  $v_i \prec v_{i+1}$. The first segment is vertical in the $2$-staircase and horizontal in the $4$-staircase. Points of $S$ are called the vertices of the staircase.

Any point $p$ in the plane defines four open axis-aligned quadrants $Q_i(p)$, $i=1,2,3,4$, as follows (see Figure~\ref{fig:Qi}):
	$Q_1(p) ~=~ \{q\in\mathbb{R}^2\mid p\prec q\}$,\label{page6}
	$Q_2(p) ~=~ \{q\in\mathbb{R}^2\mid q\prec' p\}$,\label{page7}
	$Q_3(p) ~=~ \{q\in\mathbb{R}^2\mid q\prec p\}$,\label{page8} and
    $Q_4(p) ~=~ \{q\in\mathbb{R}^2\mid p\prec' q\}$\label{page9}.
Given a point set $P$ in the plane, for $i=1,2,3,4$, let
\[
	M_i(P)=\bigcup_{p\in P} \overline{Q_i(p)},\label{page10}
\]
where $\overline{Q_i(p)}$ denotes the closure of~$Q_i(p)$.

\begin{definition}
Given a point set $P$ in the plane, the elements of $P$ that belong to the boundary of $M_i(P)$, are called the (rectilinear) $i$-{\em extremal} points of $P$ (see Figure~\ref{fig:Mi}).
\end{definition}

Note that the $i$-extremal points of $P$ are the vertices of a $i$-staircase connecting all of them. This $i$-staircase, that we call the {\em $i$-staircase associated with $P$}, is the part of the boundary of $M_i(P)$ that connects all the $i$-extremal points of $P$ (see Figure~\ref{fig:Mi}).

\begin{definition}
Given a point set $P$ in the plane, for every $J\subseteq\{1,2,3,4\}$, we say that $p\in P$ is \emph{$J$-extremal} if $p$ is $j$-extremal for every $j\in J$.
\end{definition}

\begin{definition}
Given a point set $P$ in the plane, the \emph{rectilinear convex hull} of $P$ is the set\footnote{The notation $\mathcal{RH}(P)$ is also used for the rectilinear convex hull~\cite{alegria2017}.}

\[
	RCH(P) ~=~ \bigcap_{i=1,2,3,4} M_i(P),\label{page11}
\]
\end{definition}

Figure~\ref{fig:examples1} shows some examples of rectilinear convex hulls. The boundary of $RCH(P)$ is (a part of) the union of the $1$-, $2$-, $3$- and $4$-staircases associated with $P$. Observe that the endpoints of these four staircases are $a,b,c$ and $d$, $a$ is $\{1,4\}$-extremal, $b$ is $\{1,2\}$-extremal, $c$ is $\{2,3\}$-extremal, and $d$ is $\{3,4\}$-extremal. In Figure~\ref{fig:RCH3}, as $a=d$, then $a$ is $\{1,3,4\}$-extremal and the $4$-staircase associated with $P$ consists of only point $a$.

Also observe that $RCH(P)$ is disconnected when either the intersection of the complements $\mathbb{R}^2\setminus M_1(P)$ and $\mathbb{R}^2\setminus M_3(P)$ is not empty, as shown in Figure~\ref{fig:RCH3}, or the intersection of the complements $\mathbb{R}^2\setminus M_2(P)$ and $\mathbb{R}^2\setminus M_4(P)$ is not empty. In other words, when either the $1$- and $3$-staircases associated with $P$ cross or the $2$- and $4$-staircases associated with $P$ cross.

\begin{definition}
Given a point set $P$ in the plane, a \emph{pinched point} $u$ of $RCH(P)$ occurs when $u$ is either $\{1,3\}$-extremal, as shown in Figure~\ref{fig:RCH3}, or $\{2,4\}$-extremal.
\end{definition}

\begin{definition}
Given a point set $P$ in the plane, the \emph{size} of $RCH(P)$ is the number of points of~$P$ which are $i$-extremal for at least one $i\in\{1,2,3,4\}$.
\end{definition}

From the definition of the staircases for $P$, the following observation is straightforward.

\begin{observation}\label{obs:empty}
Assume that the concatenation of the four $i$-staircases associated with $P$ is traversed counter-clockwise.
For two consecutive $i$-extremal points $p$ and $p'$, $Q_{i+2}(o)$ contains no element of $P$, where $i+2$ is taken modulo $4$ and $o=(p'_x,p_y)$ for $i=1,3$ or $o=(p_x,p'_y)$ for $i=2,4$.
\end{observation}

\begin{definition}
Given two points $u\ne v$ in the plane and a point set $P$, $B(u,v)$\label{page12} and $P(u,v)=P\cap B(u,v)$\label{page13} will denote the smallest open axis-aligned rectangle containing $u$ and $v$, and the set of points in $P$ that belong to $B(u,v)$, respectively. If $u=v$, then we define $B(u,u)$ as point $u$.
\end{definition}

Note that $u$ and $v$ are two opposed corners of $B(u,v)$.

\begin{definition}
Given a point set $P$ in the plane, we say that $RCH(P)$ is \emph{vertically separable} if rectangles $B(a,d)$ and $B(b,c)$ are separated by a vertical line. The two examples shown in Figure~\ref{fig:examples1} are vertically separable.
\end{definition}

\begin{figure}[ht]
	\centering
	    \includegraphics[scale=0.95,page=34]{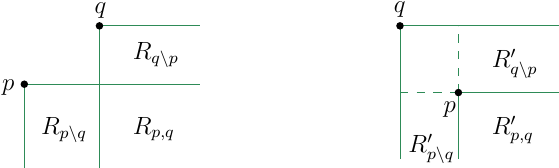} \qquad \qquad \qquad \qquad \includegraphics[scale=0.95,page=10]{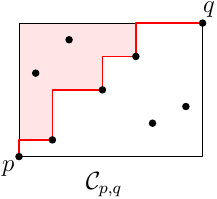}
	\caption{\small{Left: Sets $R_{p\setminus q}$,
		$R_{q\setminus p}$, $R_{p,q}$, $R'_{p\setminus q}$,
		$R'_{q\setminus p}$ and $R'_{p,q}$.} Right: Example of $\mathcal{C}_{p,q}$.}\label{fig:regions}
\end{figure}

Given a point set $S$, and a horizontal line $\ell$, let $S'$ be the image of $S$ under a reflection around $\ell$. The following lemma is key for our algorithms:

\begin{lemma}\label{lem:canonical}
Let $P$ be a point set in the plane. For all $S\subseteq P$, $|S|\ge 2$, either $RCH(S)$ or $RCH(S')$ is vertically separable.
\end{lemma}

\begin{proof}
Note that $d_x < b_x$ is necessary and sufficient for vertical separability of $B(a,d)$ and $B(b,c)$, that is, $RCH(S)$ is vertically separable. Suppose then that $b_x < d_x$, and let $\ell$ be a horizontal line. It is straightforward to see that, if we reflect the point set $S$ around $\ell$, then $S$ becomes $S'$ and we have that $RCH(S')$ is vertically separable.
\end{proof}

In each of the problems \prob{MaxRCH}, \prob{MaxEmptyRCH}, \prob{MaxAreaRCH}, and \prob{MaxWeightRCH}, we will assume that the optimal subset $S\subseteq P$ is such that $RCH(S)$ is vertically separable.

To finish this section, we give one more definition.

\begin{definition}
Given a point set $P$ in the plane, for every $p,q \in P$ such that $p\prec q$, we define $R_{p\setminus q}$\label{page14}, $R_{q\setminus p}$\label{page15}, and $R_{p,q}$\label{page16} as the subsets of $P$ in the regions $Q_4(p)\setminus \overline{Q_4(q)}$, $Q_4(q)\setminus \overline{Q_4(p)}$, and $Q_4(p)\cap Q_4(q)$, respectively (see Figure~\ref{fig:regions}, left). For every $p,q\in P$ such that $q\prec'p$, we define $R'_{p\setminus q}$\label{page17}, $R'_{q\setminus p}$\label{page18} and $R'_{p,q}$\label{page19} as the subsets of $P$ in the regions $Q_4(q)\cap Q_3(p)$, $Q_4(q)\cap Q_1(p)$ and $Q_4(p)$, respectively.
\end{definition}

Observe that if $r \in R_{p\setminus q}$ then $r \prec q$, if $r \in R_{q\setminus p}$ then $p \prec r$, and if $r \in R_{p,q}$ then $r \not\prec q$ and $p \not \prec r$.

\section{Rectilinear convex hull of maximum size}\label{sec:MaxRCH}

In this section, we solve the \prob{MaxRCH} problem. Given $P$, our goal is to combine four staircases in order to obtain a subset $S$ of $P$ whose rectilinear convex hull is of maximum size. All of this has to be done carefully, since the occurrence of pinched points may lead to overcounting.

Our algorithm to solve the \prob{MaxRCH} problem proceeds in three steps: In the first step, we compute the $2$-staircases of maximum size for every $p,q\in P$ such that $p\prec q$. In the second step, we compute what we call a \emph{triple $1$-staircase} and a \emph{triple $3$-staircase} of maximum sizes (yet to be defined). In the third and last step, we show how to combine a triple $1$-staircase and a triple $3$-staircase to solve the \prob{MaxRCH} problem. In this step, we will make sure that the solution thus obtained is vertically separable. Our algorithm will run in $O(n^3)$ time and $O(n^2)$ space. We describe now in detail the steps of our algorithm.

\begin{figure}[ht]
	\centering
	\includegraphics[scale=0.9,page=11]{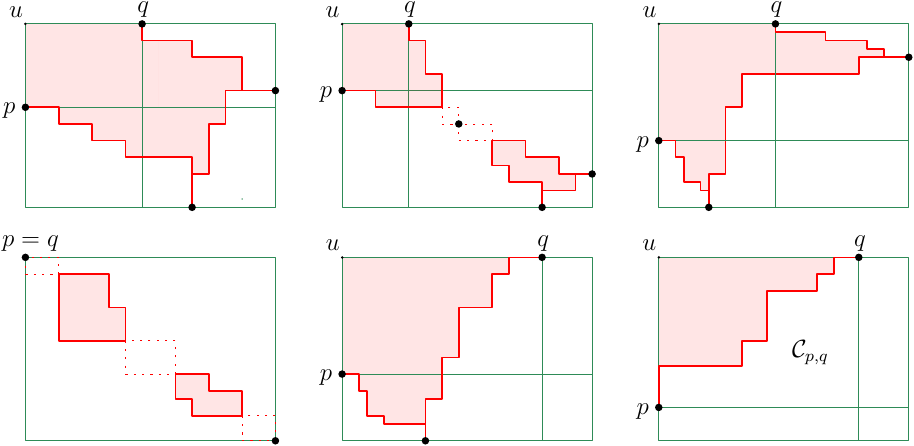}
	\caption{\small{Examples of triple $1$-staircases $\mathcal T_{p,q}$.}}\label{fig:T}
\end{figure}

{\bf The first step:} For every $p,q\in P$ such that $p\prec q$ or $p=q$, let $\mathcal C_{p,q}$\label{page20} be a $2$-staircase with endpoints $p$ and $q$ of maximum size, see Figure~\ref{fig:regions}, right. Let $C_{p,q}$\label{page21} be the number of elements of $P$ in $\mathcal C_{p,q}$. Note that $C_{p,q}$ equals the maximum number of $2$-extremal points over all $S\cup \{p,q\}$ with $S\subseteq P(p,q)$. We can easily calculate $C_{p,q}$, for all $p,q \in P$ with $p\prec q$ or $p=q$, in $O(n^3)$ time and $O(n^2)$ space, using dynamic programming with the following recurrence:
\begin{equation}\label{eq1}
    C_{p,q} ~=~ \left\{
            \begin{array}{ll}
                1 & \mathbf{if~} p=q \\
                \max\{1+C_{r,q}\}~\text{over~all}~r\in P(p,q)\cup \{q\} &
	                \mathbf{if~} p\neq q. \\
            \end{array}
        \right.
\end{equation}
Using the elements $C_{p,q}$, it is a routine matter to determine $\mathcal C_{p,q}$, for any $p\prec q$.\footnote{We note that using not so trivial methods, we can calculate all of the $C_{p,q}$'s in $O(n^2\log n)$ time. However, this yields no improvement on the overall time complexity of our algorithms.}

\begin{definition}
Given a point set $S\subseteq P$, we define the \emph{triple $1$-staircase (resp., triple $3$-staircase)} associated with $S$ as the concatenation of the $1$-, $2$- and $3$-staircases (resp., the $3$-, $4$- and $1$-staircases) associated with $S$.
\end{definition}

{\bf The second step:} In this step, our goal is to obtain a triple $1$-staircase and a triple $3$-staircase of maximum cardinality, starting and ending at some pairs of points of $P$. Triple staircases allow us to conveniently manage pinched points and disconnections of the rectilinear convex hull. Notice that the boundary of $M_1(S)\cap M_2(S)\cap M_3(S)$ (except for its two infinite rays) always belongs to the triple $1$-staircase associated with $S$.


\begin{definition}
Consider $p,q\in P $ such that $p\prec q$ or $p=q$. We use $Z(p,q)=Q_4(u)$\label{page22} to denote the fourth quadrant associated with $u$, where $u=(p_x,q_y)$, and $z(p,q)=Z(p,q)\cap P$\label{page23} to denote the points of $P$ in this quadrant.
\end{definition}

Let $\mathcal T_{p,q}$\label{page24} be the triple $1$-staircase of maximum cardinality among all subsets $S\cup \{p,q\}$ with $S\subseteq z(p,q)$. If $S'\subseteq z(p,q)$ is the set associated with $\mathcal T_{p,q}$, observe that $M_1(S'\cup \{p,q\})\cap M_2(S'\cup \{p,q\})\cap M_3(S'\cup \{p,q\})$ may contain points in $P(p,q)$, it may be disconnected, and it may have pinched points (see Figure~\ref{fig:T}). Note that $p$ and $q$ are always the endpoints of $\mathcal T_{p,q}$. Let $X_{p,q}$\label{page25} denote the set of extreme vertices of $\mathcal T_{p,q}$ (that is, the set of $1$-, $2$- and $3$-extremal points of $S'\cup \{p,q\}$), and let $T_{p,q}$\label{page26} be the cardinality of $X_{p,q}$.

\begin{figure}[ht]
	\centering
	\includegraphics[scale=0.9,page=35]{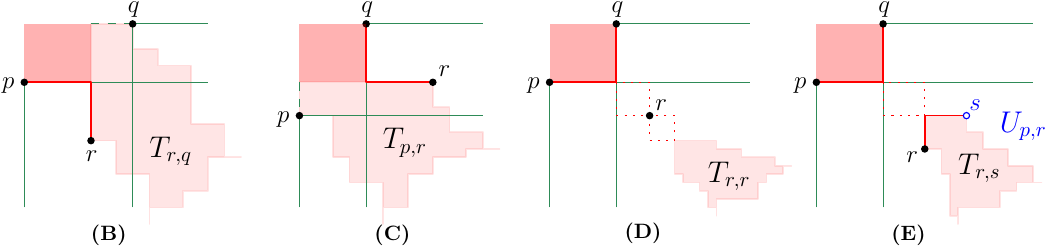}
	\caption{\small{Cases in the recursive computation of $T_{p,q}$.}}\label{fig:regions-cases}
\end{figure}

We calculate all of the $T_{p,q}$'s by dynamic programming using Equations~(\ref{eq2}) and~(\ref{eq3}). We store all of the $T_{p,q}$'s in a table $T$\label{page27}. If $\alpha_{p,q}=1$ when $p=q$, and $\alpha_{p,q}=2$ when $p\neq q$, then:
\begin{equation}\label{eq2}
     \arraycolsep=1.4pt\def\arraystretch{1.2}
     T_{p,q} ~=~ \max\left\{
         \begin{array}{lr}
             C_{p,q} & \textbf{\small (A)} \\
             1+T_{r,q}~\text{over~all}~ r\in R_{p\setminus q}  & \textbf{\small (B)}\\
             1+T_{p,r}~\text{over~all}~ r\in R_{q \setminus p}  & \textbf{\small (C)} \\
             \alpha_{p,q}+T_{r,r}~\text{over~all}~ r \in R_{p,q}  &~~\textbf{\small (D)} \\
             \alpha_{p,q}+U_{p,r}~\text{over~all}~ r \in R_{p,q}  & \textbf{\small (E)} \\
         \end{array}
     \right.
\end{equation}
where for every pair $p,r\in P$ such that $p\prec ' r$
\begin{equation}\label{eq3}
    U_{p,r} ~=~ \max\{T_{r,s}\}~\text{over~all}~ s\in R'_{p\setminus r}.
\end{equation}
Values $U_{p,r}$ are stored in a table $U$\label{page29}. The next lemma shows the correctness of this recurrence.

\begin{lemma}\label{lem:step2}
The recurrence~(\ref{eq2}) correctly calculates $T_{p,q}$, the size of $X_{p,q}$, in $O(n^3)$ time and $O(n^2)$ space.
\end{lemma}

\begin{proof}
Let $\mathcal T_{p,q}$ be an optimal triple $1$-staircase for a pair of points $p,q \in P$ such that $p \prec q$ and let $S'\subseteq z(p,q)$ be the point set associated with $\mathcal T_{p,q}$. In the counter-clockwise traversal of the triple $1$-staircase, let $p^-$ and $q^-$ be the elements of $P$ that follow and precede $p$ and $q$, respectively. Hence, $\mathcal T_{p,q}$ can be obtained as an extension of $\mathcal T_{p,q^-}$, $\mathcal T_{p^-,q}$, or $\mathcal T_{p^-,q^-}$.

If $p^-,q^- \in B(p,q)$, then necessarily $\mathcal T_{p,q}$ is a $2$-staircase (the $1$ and $3$-staircases of $\mathcal T_{p,q}$ consist of only points $p$ and $q$, respectively), so $T_{p,q}=C_{p,q}$ and case $\textbf{\small (A)}$ is used to set $T_{p,q}$. Thus, we assume in the rest of the proof that at least one of $p^-$ and $ q^-$ is not in $B(p,q)$. See Figure~\ref{fig:regions-cases} for the cases.

If $p^-\in R_{p\setminus q}$, then we have $p^- \prec q$ and $p^-$ is not 3-extremal in $\mathcal T_{p^-,q}$. We use case $\textbf{\small (B)}$ to find a point $r$ that plays the role of $p^-$ to compute the value of $T_{p,q}$. If $q^- \in R_{q \setminus p}$, then we have $p\prec q^-$ and $q^-$ is a point that is not 1-extremal in $\mathcal T_{p,q^-}$. We use case~$\textbf{\small (C)}$ to find a point $r$ that plays the role of $q^-$ to compute the value of $T_{p,q}$.

Suppose now that $p^-$ is in $R_{p,q}$ (a similar reasoning applies if $q^-$ is in $R_{p,q}$). In this case, $q^-$ cannot be in $B(p,q)$ and if $q^- \in R_{q \setminus p}$ then we use case~$\textbf{\small (C)}$. If $q^-\in R_{p,q}$, there are two cases to analyze: $p^-=q^-$ and $p^-\neq q^-$.
If $p^- = q^-$, then we can use case $\textbf{\small (D)}$ to find a point $r$ that plays the role of $p^-$ to compute the value of $T_{p,q}$. When $p^- \neq q^-$, we prove that $p^- \prec q^-$, and then case $\textbf{\small (E)}$ can be used to find a pair of points $r$ and~$s$ playing the roles of $p^-$ and $q^-$, both in $R_{p,q}$, with maximum value $T_{r,s}$. By Observation~\ref{obs:empty}, as $p$ and $p^-$ are consecutive $1$-extremal points, then $Q_3(o)$ contains no element in $S'$, where $o=(p^-_x,p_y)$. In particular, since $q^-$ is in $R_{p,q}$, this implies that $q^-$ cannot be in either $Q_2(p^-)$ or $Q_3(p^-)$. Moreover, as $q$ and $q^-$ are consecutive $3$-extremal points, then $Q_1(o')$ contains no element in $S'$ again by Observation~\ref{obs:empty}, where $o'=(q_x,q^-_y)$. As a consequence, $q^-$ cannot be in $Q_4(p^-)$, so we conclude that $q^- \in Q_1(p^-)$.

To compute tables $T$ and $U$, we scan the elements of $P$ from right to left. Each time an element $p\in P$ is encountered, we scan all of the $q\in P$ such that $p_x < q_x$, again from right to left. When $p\prec q$ we compute $T_{p,q}$, and when $p\prec'q$ we compute $U_{p,q}$. Each entry of $T$ and $U$ is determined in $O(n)$ time. Thus, $U$ and $T$ can be computed in overall $O(n^3)$ time and $O(n^2)$ space.
Cases $\textbf{\small (A)}$ to $\textbf{\small (D)}$ are in $O(n)$. We charge the work done in case $\textbf{\small (E)}$ to constructing table $U$, which can be done in $O(n)$ time per entry. Thus the entire complexity is in $O(n^3)$ time and $O(n^2)$ space.
\end{proof}

In a totally analogous way, we can calculate triple $3$-staircases of maximum size. For $p\prec q$ or $p=q$, let $T'_{p,q}$\label{page31} be the size of the triple $3$-staircase $\mathcal T'_{p,q}$\label{page30} of maximum cardinality among all subsets $S'\cup \{p,q\}$, where $S'$ is now a subset of points of $P$ in the quadrant $Q_2(v)$ with $v=(q_x,p_y)$. After rotating the coordinates by $\pi$, observe that the triple $3$-staircase $\mathcal T'_{p,q}$ is the triple $1$-staircase $\mathcal T_{q,p}$. Thus, by symmetry with the $T_{p,q}$'s, all the $T'_{p,q}$'s can also be calculated in $O(n^3)$ time and $O(n^2)$ space.

{\bf The third step:} In this step, we show how to combine a triple $1$-staircase and a triple $3$-staircase to solve the \prob{MaxRCH} problem. Recall that the solution must be vertically separable. Next we give the definition of a $4$-separator and then we show that it is equivalent to vertical separability.

\begin{definition}
Let $S\subseteq P$ be any subset with $|S|\ge 2$. Given four (not necessarily distinct) extremal points $p,q,r,s\in S$, we say that the tuple $(p,q,r,s)$ is a \emph{$4$-separator} of $RCH(S)$ if the following five conditions are satisfied: $(1)$ $p\prec q$ or $p=q$; $(2)$ $q \prec' r$; $(3)$ $r\prec s$ or $r=s$; $(4)$ $p$ and $r$ are consecutive points in the $1$-staircase of $S$; and $(5)$ $s$ and $q$ are consecutive points in the $3$-staircase of $S$ (see Figure~\ref{fig:separators1}).
\end{definition}

\begin{figure}[ht]
	\centering
	\includegraphics[scale=0.9,page=31]{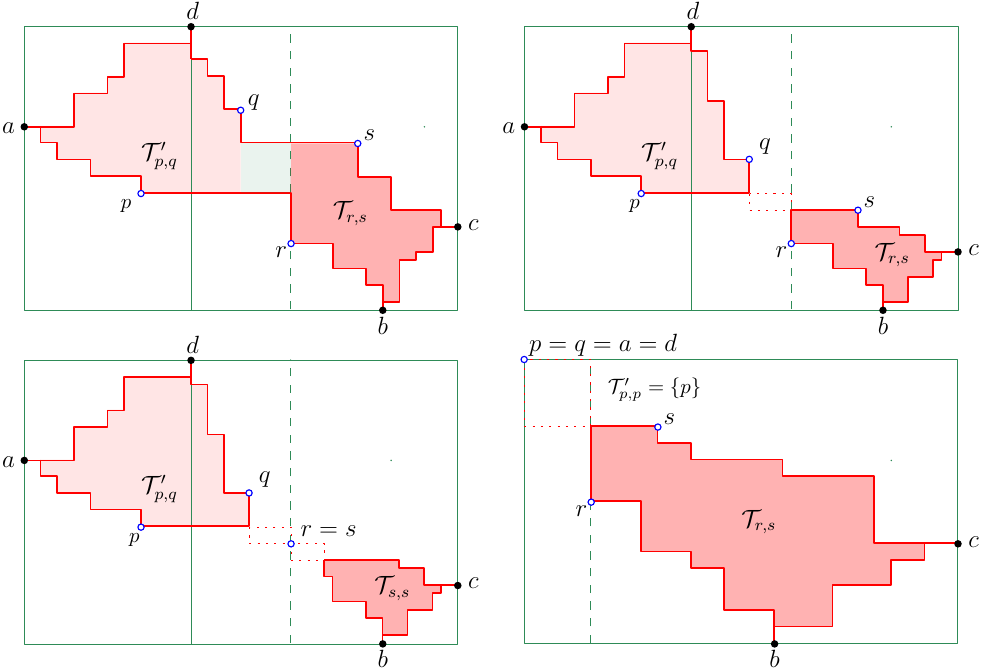}
	\caption{\small{Examples of $4$-separators $(p,q,r,s)$.}}\label{fig:separators1}
\end{figure}

The following lemma shows the equivalence between the existence of $4$-separators and vertical separability.

\begin{lemma}\label{lemm:4separator}
$RCH(S)$ is vertically separable if and only if $RCH(S)$ has a $4$-separator.
\end{lemma}

\begin{proof}
Let us first assume that $RCH(S)$ is vertically separable, that is, rectangles $B(a,d)$ and $B(b,c)$ are separated by a vertical line. Recall that $a$, $b$, $c$, and $d$ are the leftmost, bottommost, rightmost, and topmost points of $S$, respectively. Then, we can argue the following: If $a\prec d$, then $S$ has at least one $1$-extremal point to each side of the vertical line through $d$. Otherwise, if $a=d$, then $S$ has at least one $1$-extremal point to the right side. Thus, covering both cases, let $p$ and $r$ be the two consecutive $1$-extremal points of $S$ such that $p_x\le d_x < r_x$. Now, given that $d_x<r_x$, we have that: If $r$ is also $3$-extremal, thus a pinched point, then $S$ has at least one $3$-extremal point to the left side of the vertical line through $r$. Otherwise, if $r$ is not $3$-extremal, then $S$ has at least one $3$-extremal point to each side of this line. Thus, we can define $s$ and $q$ as the two consecutive $3$-extremal points of $S$ such that $q_x < r_x\le s_x$. It is straightforward to see now that $(p,q,r,s)$ is a $4$-separator of $RCH(S)$. Note that, when $p=a=d$, then necessarily $q=p=a=d$.

Assume now that $(p,q,r,s)$ is a $4$-separator of $RCH(S)$. We then have that: $d\prec'q$ or $d=q$, and $r\prec'b$ or $r=b$. These conditions, together with $q\prec' r$, directly imply that rectangles $B(a,d)$ and $B(b,c)$ are separated by a vertical line, thus $RCH(S)$ is vertically separable.
\end{proof}

Using $4$-separators, we show how to find an optimal solution that is vertically separable. Among all subsets $S$ of $P$ such that $RCH(S)$ is vertically separable, let $S_0$ be a subset of $P$ such that $RCH(S_0)$ has maximum size. Let $(p,q,r,s)$ be a $4$-separator of $RCH(S_0)$. The key observation is that the vertices of $\mathcal{T}'_{p,q}\cup \mathcal{T}_{r,s}$ are the set of extremal points of $S_0$. Note that $\mathcal{T}'_{p,q}\cap \mathcal{T}_{r,s}=\emptyset$ and $|RCH(S_0)|$, the size of $RCH(S_0)$, is $T'_{p,q}+T_{r,s}$.

Thus, we proceed as follows: For given $p,s\in P$ such that $p_x<s_x$, let $\mathcal{S}_{p,s}$\label{page32} be the rectilinear convex hull of maximum size, among all subsets $S\subseteq P$ containing $p$ and $s$ such that there exist two points $q,r\in S$ with $(p,q,r,s)$ being a $4$-separator of $RCH(S)$. Let $S_{p,s}$\label{page33} be the size of $\mathcal{S}_{p,s}$. Note that $S_{p,s} = T'_{p,q}+T_{r,s}$ for some $4$-separator $(p,q,r,s)$. Then, the following equations allow us to calculate $|RCH(S_0)|$ in $O(n^3)$ time and $O(n^2)$ space, as Theorem~\ref{the:MaxRCH} proves:
\begin{equation}\label{eqopt1}
    |RCH(S_0)| ~=~ \max\{S_{p,s}\}~\text{over~all}~ p,s\in P ~\text{such that}~ p_x < s_x
\end{equation}
where for each pair of points $p,s\in P$ with $p_x < s_x$
\begin{equation}\label{eqopt2}
    S_{p,s} ~=~ \max\{T'_{p,q}+T_{r,s}\}~\text{over~all}~ q,r\in P ~\text{such that $(p,q,r,s)$ is a $4$-separator}.
\end{equation}

\begin{theorem}\label{the:MaxRCH}
    The \prob{MaxRCH} problem can be solved in $O(n^3)$ time and $O(n^2)$ space.
\end{theorem}

\begin{proof}
According to Equations~(\ref{eqopt1}) and~(\ref{eqopt2}), we only need to show how to compute $S_{p,s}$ in linear time, for given $p$ and $s$. Let $Q_p$ be the set of all points $q\in P$ such that $q\prec's$, and $p\prec q$ or $p=q$. Let $Q_s$ be the set of all points $r\in P$ such that $p\prec'r$, and $r\prec s$ or $r=s$. Note that $Q_p\cap Q_s=\emptyset$ and that $p$ and $s$ belong to $Q_p$ and $Q_s$, respectively, only when $p_y > s_y$. Let $L_{p,s}$ be the list of the elements of $Q_p\cup Q_s$ sorted by $x$-coordinate. Observe that if $(p,q,r,s)$ is a $4$-separator of $RCH(S_{p,s})$, then $r\in Q_s$, $q\in Q_p$, and $q$ is the point $q^*\in Q_p$ from the beginning of $L_{p,s}$ to $r$ such that $T'_{p,q^*}+T_{r,s}=\max\{T'_{p,q'}+T_{r,s}\}$ over all $q'\in Q_p$ from the beginning of $L_{p,s}$ to $r$.

We calculate $S_{p,s}$ by processing the elements of $L_{p,s}$ in order. For an element $t$ in $L_{p,s}$, let $q_t^*$ be the point in $Q_p$ maximizing $T'_{p,q'}$ over all $q'\in Q_p$ from the beginning of $L_{p,s}$ to $t$ (including $t$). When processing a point $t$, observe that if $t\in Q_p$, then $q_t^*$ is either $t$ or $q_{t-1}^*$. Otherwise, if $t\in Q_s$, then $q_t^*=q_{t-1}^*$. Moreover, if $t\in Q_s$, then we set $q=q_t^*$ and $r=t$, and consider $(p,q,r,s)$ as a \emph{feasible} $4$-separator of $\mathcal{S}_{p,s}$. After processing the last element of $L_{p,s}$, among all the (linear number) feasible separators, we return the solution $S_{p,s}$ induced by the feasible separator that maximizes $T'_{p,q}+T_{r,s}$. Thus, $S_{p,s}$ can be calculated in $O(n)$ time, once tables $T$ and $T'$ have been constructed.
\end{proof}

\section{Maximum size/area empty rectilinear convex hulls and maximum weight rectilinear convex hull}\label{sec:EmptyRCH}

In this section, we show how to adapt the algorithm of Section~\ref{sec:MaxRCH} to solve the \prob{MaxEmptyRCH}, the \prob{MaxAreaRCH} and the \prob{MaxWeightRCH} problems. The first observation is that, as the optimal solution for any of these problems is the rectilinear convex hull of a subset $S$ of $P$, then Lemmas~\ref{lem:canonical} and~\ref{lemm:4separator} hold. This implies that we can assume that rectangles $B(a,d)$ and $B(b,c)$ are separated by a vertical line, so a 4-separator exists for the optimal solution in any of the problems. As a consequence, the algorithms to solve these three problems follow the same scheme as the algorithm described in the previous section, and we only need to show how to adapt in each problem the calculation of the $2$-staircases, the triple $1$- and $3$-staircases and the rectilinear convex hulls $\mathcal{S}_{p,s}$ to fulfill the requirements on emptiness, area or weight.

We start by solving the \prob{MaxEmptyRCH} problem, we continue with the \prob{MaxAreaRCH} problem and we finish with the \prob{MaxWeightRCH} problem.

\subsection{Maximum size empty rectilinear convex hull}\label{sec:EmptyRectCH_1}

To solve the \prob{MaxEmptyRCH} problem in $O(n^3)$ time and $O(n^2)$ space, we modify the steps of our previous algorithm. These modifications ensure that the ``interiors'' of the triple $1$- and $3$-staircases and the rectilinear convex hulls are empty. Recall that in this problem we are looking for a subset $S\subseteq P$ such that $RCH(S)$ has maximum size and there is no element of $P$ in the interior of $RCH(S)$.

{\bf The first step:} For a pair of points $p,q\in P$ such that $p\prec q$ or $p=q$, we say that the $2$-staircase associated with a subset $S$ of $P(p,q)$ is {\em empty} (recall that $P(p,q)$ is the set of points in $P$ that belong to the rectangle $(B(p,q)$) if no point of $P$ is in the interior of $B(p,q)\cap M_2(S\cup \{p,q\})$, see Figure~\ref{fig:Ci-empty}, left.

\begin{figure}[ht]
	\centering
    \includegraphics[scale=0.9,page=40]{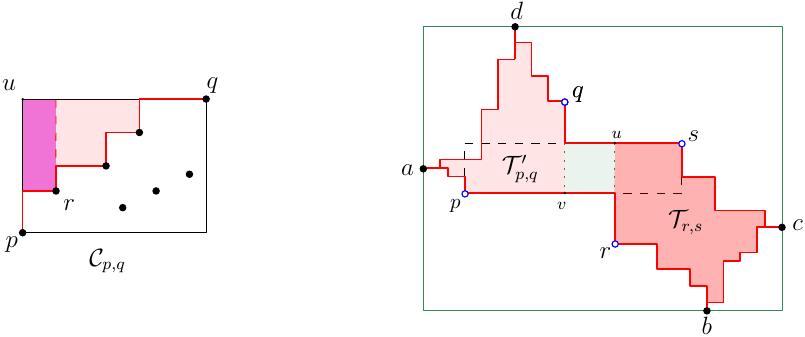}
	\caption{\small{Left: Example of $\mathcal{C}_{p,q}$ when the $2$-staircases must be empty. Right: Third step of the algorithm when the interior must be empty of elements of $P$.}}\label{fig:Ci-empty}
\end{figure}

Let $\mathcal C_{p,q}$\label{page34} be the empty $2$-staircase of maximum cardinality over all subsets $S\cup \{p,q\}$ with $S\subseteq P(p,q)$, and let $C_{p,q}$ be the size of $\mathcal C_{p,q}$. Observe that if $u$ is the point $(p_x,q_y)$ and $r\in P$ is the vertex of the $2$-staircase that follows $p$, then $P(r,u)=\emptyset$. Thus, values $C_{p,q}$ can be computed using the following recurrence:
\begin{equation}\label{eq1b}
    C_{p,q} ~=~ \left\{
            \begin{array}{ll}
                1 & \mathbf{if~} p=q \\
                \max\{1+C_{r,q}\}~\text{over~all}~r\in P(p,q)\cup \{q\}~\text{such~that}~P(r,u)=\emptyset &
	                \mathbf{if~} p\neq q. \\
            \end{array}
        \right.
\end{equation}
As $q$ and $u$ are on the same horizontal line, $P(q,u)$ is not defined. In this case, we assume that $P(q,u)$ is the empty set. Using standard range counting techniques~\cite{preparata2012computational}, we can preprocess the grid $G$ generated by the vertical and horizontal lines through the elements of $P$ in $O(n^2)$ time and space, so that for every pair of vertices $u,v$ of $G$ we can query the number of points in $P(u,v)$ in $O(1)$ time. Thus, we can decide whether $P(u,v)=\emptyset$ in $O(1)$ time. Therefore, values $C_{p,q}$ can be calculated in $O(n^3)$ time and $O(n^2)$ space.

{\bf The second step:} For every $p,q\in P$ such that $p\prec q$ or $p=q$, we say that the triple $1$-staircase $\mathcal{T}$ corresponding to a subset $S$ of $z(p,q)$ is \emph{empty} (recall that $z(p,q)$ is the set of points in $P$ that belong to $Z(p,q)$, the fourth quadrant associated with point $(p_x,q_y)$) if the (disconnected) region $\mathcal{O}_{\mathcal{T}}=\overline{Z(p,q)}\cap M_1(S\cup \{p,q\})\cap M_2(S\cup \{p,q\})\cap M_3(S\cup \{p,q\})$\label{page35} associated with $\mathcal{T}$ contains no element of $P$. Please, refer to Figure~\ref{fig:T}, where the shaded areas correspond to $\mathcal{O}_{\mathcal{T}}$. Let $\mathcal T_{p,q}$\label{page36} be the empty triple $1$-staircase of maximum size among all subsets $S\cup \{p,q\}$, with $S\subseteq z(p,q)$. Let $E(p,q,r)$\label{page37} denote the interior of $B(p,q)\cap Q_2(r)$ if $p\ne q$, and the empty set if $p=q$. In Figure~\ref{fig:regions-cases} examples of open rectangles $E(p,q,r)$ are shown as shaded rectangles . We show how to compute $T_{p,q}$\label{page38}, the cardinality of the set of extreme vertices of $\mathcal T_{p,q}$, using the following equations that are similar to Equations~(\ref{eq2}) and~(\ref{eq3}):
\begin{equation}\label{eq2b}
     \arraycolsep=1.4pt\def\arraystretch{1.2}
     T_{p,q} ~=~ \max\left\{
         \begin{array}{lr}
             C_{p,q} & \textbf{\small (A)} \\
             1+T_{r,q}~\text{over~all}~ r\in R_{p\setminus q}~:~P\cap E(p,q,r)=\emptyset  & \textbf{\small (B)}\\
             1+T_{p,r}~\text{over~all}~ r\in R_{q \setminus p}~:~P\cap E(p,q,r)=\emptyset   & \textbf{\small (C)} \\
             \alpha_{p,q}+T_{r,r}~\text{over~all}~ r \in R_{p,q}~:~P\cap E(p,q,r)=\emptyset &~~\textbf{\small (D)} \\
             \alpha_{p,q}+U_{p,r}~\text{over~all}~ r \in R_{p,q}~:~P\cap E(p,q,r)=\emptyset & \textbf{\small (E)} \\
         \end{array}
     \right.
\end{equation}
where for every pair $p,r\in P$ such that $p\prec' r$
\begin{equation}\label{eq3b}
    U_{p,r} ~=~ \max\{T_{r,s}\}~\text{over~all}~ s\in R'_{p\setminus r}.
\end{equation}

In case $\textbf{\small (A)}$, $\mathcal{O}_{\mathcal T_{p,q}}$ is empty as $\mathcal C_{p,q}$ is an empty 2-staircase.  Equation~\eqref{eq2b} is obtained from Equation~\eqref{eq2} by further constraining $r$ in the cases from $\textbf{\small (B)}$ to $\textbf{\small (E)}$ to satisfy $P\cap E(p,q,r)=\emptyset$. This guarantees that the interior of $\mathcal{O}_{p,q}$ is empty of elements of $P$ for all $p,q$. Verifying that $P\cap E(p,q,r)=\emptyset$ can be decided in $O(1)$ time by using a range counting query. The proof of correctness of Equations~(\ref{eq2b}) and~(\ref{eq3b}) follows the same steps as in Lemma~\ref{lem:step2}. Hence, computing the new table $T$ can be done in $O(n^3)$ time and $O(n^2)$ space. By symmetry, values in $T'_{p,q}$, the sizes of the empty triple $3$-staircases $\mathcal T'_{p,q}$\label{page39}, can also be calculated in $O(n^3)$ time and $O(n^2)$ space.

{\bf The third step:} For given $p,s\in P$ such that $p_x<s_x$, let $\mathcal{S}_{p,s}$\label{page40} be the \emph{empty} rectilinear convex hull of maximum size, among all subsets $S\subseteq P$ containing $p$ and $s$ such that $RCH(S)$ is empty and there exist two points $q,r\in S$ with $(p,q,r,s)$ being a $4$-separator of $RCH(S)$. Let $S_{p,s}$ denote the size of $\mathcal{S}_{p,s}$. To compute $S_{p,s}$ we have to distinguish whether $p\prec' s$ or $p\prec s$. If $p\prec' s$ (see Figure~\ref{fig:separators1}, top-right, or any in the bottom), then $S_{p,s}=\max \{ T'_{p,q}+T_{r,s}\}$ over all $4$-separators $(p,q,r,s)$. Otherwise, if $p\prec s$ (see Figure~\ref{fig:Ci-empty}, right), then we must ensure that each $4$-separator $(p,q,r,s)$ satisfies the emptiness of the rectangle $B(u,v)$ (that is, $P(u,v)=\emptyset$), where $u=(r_x,s_y)$ and $v=(q_x,p_y)$.

If $S_0$ is a subset of $P$ such that $RCH(S_0)$ is empty, vertically separable and of maximum size, the new equations to compute $|RCH(S_0)|$ are:
\begin{equation}\label{eqopt1b}
    |RCH(S_0)| ~=~ \max\{S_{p,s}\}~\text{over~all}~ p,s\in P ~\text{with}~ p_x < s_x
\end{equation}
where for each pair of points $p,s\in P$ such that $p_x < s_x$
\begin{equation}\label{eqopt2b}
     S_{p,s} ~=~ \left\{
         \begin{array}{lr}
             \max\{T'_{p,q}+T_{r,s}\}~\text{: $(p,q,r,s)$ is a $4$-separator} & p\prec' s \\
             \max\{T'_{p,q}+T_{r,s}\}~\text{: $(p,q,r,s)$ is a $4$-separator, $P(u,v)=\emptyset$} & p\prec s. \\
         \end{array}
         \right.
\end{equation}

\begin{theorem}\label{the:MaxEmptyRCH}
The \prob{MaxEmptyRCH} problem can be solved in $O(n^3)$ time and $O(n^2)$ space.
\end{theorem}

\begin{proof}
Again, we only need to show that given points $p$ and $s$ $S_{p,s}$ can be computed in linear time. When $p\prec' s$, we argue as in the proof of Theorem~\ref{the:MaxRCH}. However, when $p\prec s$, we need to only consider $4$-separators such that $P(u,v)=\emptyset$. Let $Q_{p,s}=\{t\in P: p\prec t\prec s\}$, which satisfies $Q_{p,s}\cap Q_p=\emptyset$ and $Q_{p,s}\cap Q_s=\emptyset$. Recall that $Q_p\cap Q_s=\emptyset$ and that, when $p\prec s$, $Q_p=\{ q\in P: p\prec q, q\prec' s\}$ and $Q_s=\{r\in P: r\prec s, p\prec' r\}$. Let $L_{p,s}$ be the list of the points of $Q_p\cup Q_s\cup Q_{p,s}$ sorted by $x$-coordinate. Assuming that we have already sorted $P$ by $x$-coordinate $L_{p,s}$ is obtained in $O(n)$ time.

As before, we calculate $S_{p,s}$ by processing  the elements of $L_{p,s}$ in order. For an element $t$ in $L_{p,s}$, let $q_t^*$ be the point in $Q_p$ maximizing $T'_{p,q'}$ over all $q'\in Q_p$ from the beginning of $L_{p,s}$ to $t$ (including $t$) subject to there being no elements of $Q_{p,s}$ in $L_{p,s}$ from $q'$ to $t$. When processing a point $t\in Q_{p,s}$, we set $q_t^*=\nil$ denoting that $q_t^*$ is undefined. Observe that, when processing a point $t\in Q_p$, if $q_{t-1}^*=\nil$ then $q_t^*=t$, and if $q_{t-1}^*\ne\nil$ then $q_t^*$ is either $t$ or $q_{t-1}^*$. When processing a point $t\in Q_s$, then $q_t^*=q_{t-1}^*$, and if $q_{t-1}^*\ne \nil$, then we set $q=q_t^*$ and $r=t$, and consider $(p,q,r,s)$ as a feasible $4$-separator of $\mathcal{S}_{p,s}$. Note that for this $4$-separator we have $P(u,v)=\emptyset$. After processing all elements in $L_{p,s}$, $S_{p,s}$ is determined by a feasible $4$-separator that maximizes $T'_{p,q}+T_{r,s}$.
\end{proof}

\subsection{Maximum area empty rectilinear convex hull}\label{sec:EmptyRectCH_2}

In the \prob{MaxAreaRCH} problem, we determine an empty rectilinear convex hull of maximum area. To solve this problem, we proceed as in the previous subsection. The only difference is that we sum areas in all of our recurrences, instead of counting points. Given a bounded set $Z\subset\mathbb{R}^2$, we denote the area of $Z$ as $\area(Z)$.

Now, $\mathcal C_{p,q}$\label{page41}, $\mathcal T_{p,q}$\label{page42}, $\mathcal T'_{p,q}$\label{page43}, and $\mathcal S_{p,q}$\label{page44} are as described in Section~\ref{sec:EmptyRectCH_1}, with the difference that they maximize area instead of maximizing size. The areas are defined as follows. If $S=\{p,v_2, \ldots , v_{k-1},q\}$ is the set of vertices of an empty $2$-staircase, we define the area of this staircase as $\area(B(p,q)\cap M_2(S))$. For an empty triple $1$-staircase or an empty triple $3$-staircase $\mathcal T$, its area is the area of its associated region $\mathcal O_{\mathcal T}$. The area of a rectilinear convex hull is the area of its interior.

{\bf The first step:} For a pair of points $p,q\in P$ such that $p\prec q$ or $q=p$, we compute $C_{p,q}$, the area of $\mathcal C_{p,q}$, using the following recurrence, which is a variant of Equation~(\ref{eq1b}) maximizing area:
\begin{equation}\label{eq1area}
    C_{p,q} ~=~ \left\{
            \begin{array}{ll}
                0 & \mathbf{if~} p=q \\
                \max_{r\in P(p,q)\cup \{q\}}\{\area (B(r,u))+C_{r,q}\}~:~P(r,u)=\emptyset &
	                \mathbf{if~} p\neq q. \\
            \end{array}
        \right.
\end{equation}
where $u=(p_x,q_y)$. As $B(q,u)$ is not defined, we set $\area (B(q,u))=0$.

{\bf The second step:} For every $p,q\in P$ such that $p\prec q$ or $p=q$, let $T_{p,q}$\label{page45} be the area of $\mathcal T_{p,q}$. All $T_{p,q}$'s can be calculated in $O(n^3)$ time and $O(n^2)$ space using the following equations, which are variants of Equations~(\ref{eq2b}) and~(\ref{eq3b}) maximizing area (recall that if $p=q$, then rectangle $E(p,q,r)=\emptyset$, so its area is $0$):
\begin{equation}\label{eq2area}
     \arraycolsep=1.4pt\def\arraystretch{1.2}
     T_{p,q} ~=~ \max\left\{
         \begin{array}{lr}
             C_{p,q} & \textbf{\small (A)} \\
             \area(E(p,q,r))+T_{r,q}~\text{over~all}~ r\in R_{p\setminus q}~:~P\cap E(p,q,r)=\emptyset  & \textbf{\small (B)}\\
             \area(E(p,q,r))+T_{p,r}~\text{over~all}~ r\in R_{q \setminus p}~:~P\cap E(p,q,r)=\emptyset   & \textbf{\small (C)} \\
             \area(E(p,q,r))+T_{r,r}~\text{over~all}~ r \in R_{p,q}~:~P\cap E(p,q,r)=\emptyset   &~~\textbf{\small (D)} \\
             \area(E(p,q,r))+U_{p,r}~\text{over~all}~ r \in R_{p,q}~:~P\cap E(p,q,r)=\emptyset   & \textbf{\small (E)} \\
         \end{array}
     \right.
\end{equation}
where for every pair $p,r\in P$ such that $p\prec' r$
\begin{equation}\label{eq3area}
    U_{p,r} ~=~ \max\{T_{r,s}\}~\text{over~all}~ s\in R'_{p\setminus r}.
\end{equation}

The areas $T'_{p,q}$ for the empty triple $3$-staircases $\mathcal T'_{p,q}$ can be calculated in a similar way.

{\bf The third step:} Let $S_{p,s}$\label{page46} be the area of $\mathcal{S}_{p,s}$. Recall that, for given $p,s\in P$ such that $p_x<s_x$, $\mathcal{S}_{p,s}$ is the empty rectilinear convex hull of maximum area, among all subsets $S\subseteq P$ containing $p$ and $s$ such that $RCH(S)$ is empty and there exist two points $q,r\in S$ with $(p,q,r,s)$  a $4$-separator of $RCH(S)$. Observe that if $p\prec' s$, then $S_{p,s}=T'_{p,q}+T_{r,s}$ for some $4$-separator $(p,q,r,s)$. Otherwise, if $p\prec s$ then $S_{p,s}=\area(B(u,v))+T'_{p,q}+T_{r,s}$ for some $4$-separator $(p,q,r,s)$, subject to $P(u,v)=\emptyset$, where $u=u(r)=(r_x,s_y)$ and $v=v(q)=(q_x,p_y)$ (see Figure~\ref{fig:Ci-empty}, right).

Given that $\area(B(u,v))$ depends on both $r$ and $q$, using the inclusion/exclusion principle, we can then calculate $S_{p,s}$ as
$$S_{p,s}=T'_{p,q}+\area(B(v,s)) + T_{r,s}+\area(B(p,u)) - \area(B(p,s)).$$
Since $p$ and $s$ are fixed, note that $U(p,q,s)=T'_{p,q}+\area(B(v(q),s))$ depends only on $q$ and $V(p,r,s)=T_{r,s}+\area(B(p,u(r)))-\area(B(p,s))$ depends only on $r$. Each of these two values can be computed in $O(1)$ time, once $T$ and $T'$ have been computed in the second step. If $S_0$ is a subset of $P$ such that $RCH(S_0)$ is empty, vertically separable and of maximum area, the new equations to compute $RCH(S_0)$ are:
\begin{equation}\label{eqopt1area}
    |RCH(S_0)| ~=~ \max\{S_{p,s}\}~\text{over~all}~ p,s\in P ~\text{with}~ p_x<s_x,
\end{equation}
where for each pair of points $p,s\in P$ such that $p_x < s_x$
\begin{equation}\label{eqopt2area}
     S_{p,s} ~=~ \left\{
         \begin{array}{lr}
             \max_{\{\text{$4$-separators $(p,q,r,s)$\}}}\{T'_{p,q}+T_{r,s}\}~ & p \prec' s \\
             \max_{\{\text{$4$-separators $(p,q,r,s)$\}}}\{U(p,q,s) + V(p,r,s))\}~\text{: $P(u(r),v(q))=\emptyset$} & p\prec s.\\
         \end{array}
         \right.
\end{equation}

\begin{theorem}\label{the:MaxAreaRCH}
The \prob{MaxAreaRCH} problem can be solved in $O(n^3)$ time and $O(n^2)$ space.
\end{theorem}

\begin{proof}
The proof follows the proof of Theorem~\ref{the:MaxEmptyRCH}. The only difference is that, when processing an element $t$ in $L_{p,s}$, $q_t^*$ is the point in $Q_p$ maximizing $T'_{p,q'}+\area(B(v(q'),s))$, instead of maximizing $T'_{p,q'}$. After processing all elements in $L_{p,s}$, $S_{p,s}$ is determined by a feasible $4$-separator that maximizes $T'_{p,q}+\area(B(v,s))+T_{r,s}+\area(B(p,u))-\area(B(p,s))$.
\end{proof}

\subsection{Maximum weight rectilinear convex hull}\label{sec:EmptyRectCH_3}

In the \prob{MaxWeightRCH} problem, each input point $p$ of $P$ comes with a (positive or negative) weight $w(p)$. We determine a subset $S\subseteq P$ such that $RCH(S)$ has maximum weight, that is, such that $\sum_{p\in P\cap RCH(S)}{}w(p)$ is maximized.

The algorithm to solve this problem combines the ideas of the previous algorithms and follows the same steps, however,  now we add weights. We define $\weight(Z)=\sum_{p\in P\cap Z}w(p)$ as the weight of a region $Z\subset\mathbb{R}^2$. Using the same range counting techniques~\cite{preparata2012computational} as in Section~\ref{sec:MaxRCH}, we can preprocess the grid $G$ generated by the vertical and horizontal lines through the elements of $P$ in $O(n^2)$ time and space, so that for every pair of vertices $u,v$ of $G$ we can query $\weight(B(u,v))=\sum_{p\in P(u,v)}w(p)$ in $O(1)$ time, for any rectangle $B(u,v)$.

Now, $\mathcal C_{p,q}$\label{page47}, $\mathcal T_{p,q}$\label{page48}, $\mathcal T'_{p,q}$\label{page49}, and $\mathcal S_{p,q}$\label{page50} are as described in Section~\ref{sec:MaxRCH}, except that weight is maximized. The weights are defined as follows. If $S=\{p,v_2,\ldots,v_{k-1},q\}$ is the set of vertices of a $2$-staircase, its weight is defined as $w(p)+w(q)+\weight(B(p,q)\cap M_2(S))$. Note that the weights of all points in $S$ are included in this formulae. For a triple $1$-staircase or a triple $3$-staircase $\mathcal T$, its weight is the addition of the weights of the points of $P$ that appear on the boundary or in the interior of $\mathcal{O}_{\mathcal{T}}$, the region associated with $\mathcal{T}$. Finally, the weight of a rectilinear convex hull is the addition of the points of $P$ on the boundary or the interior of the rectilinear convex hull.

{\bf The first step:} If $C_{p,q}$\label{page51} is the weight of $\mathcal C_{p,q}$, for a pair of points $p,q\in P$ such that $p\prec q$ or $q=p$, all $C_{p,q}$'s can be computed in $O(n^3)$ time and $O(n^2)$ space using the following recurrence:
\begin{equation}\label{eq1weight}
	C_{p,q} ~=~ \left\{
		\begin{array}{ll}
			w(p) & \mathbf{if~} p=q \\
			w(p)+\max_{r\in P(p,q)\cup \{q\}}\left\{\weight(B(r,u))+C_{r,q}\right\} &
				\mathbf{if~} p\neq q. \\
		\end{array}
	\right.
\end{equation}
where $u=(p_x,q_y)$. We set $\weight(B(q,u))=0$ as $B(q,u)$ is not defined.

{\bf The second step:} If $T_{p,q}$\label{page52} is the weight of $\mathcal T_{p,q}$, for every $p,q\in P$ such that $p\prec q$ or $p=q$, then all $T_{p,q}$'s (and, by symmetry, all $T'_{p,q}$'s) can be calculated in $O(n^3)$ time and $O(n^2)$ space using the following equations, where $\alpha_{p,q}=w(p)$ if $p=q$, and $\alpha_{p,q}=w(p)+w(q)$ if $p\neq q$:
\begin{equation}\label{eq2weight}
     \arraycolsep=1.4pt\def\arraystretch{1.2}
     T_{p,q} ~=~ \max\left\{
         \begin{array}{lr}
             C_{p,q} & \textbf{\small (A)} \\
             w(p) + \weight(E(p,q,r))+T_{r,q}~\text{over~all}~ r\in R_{p\setminus q}  & \textbf{\small (B)}\\
             w(q) + \weight(E(p,q,r))+T_{p,r}~\text{over~all}~ r\in R_{q \setminus p}   & \textbf{\small (C)} \\
             \alpha_{p,q} + \weight(E(p,q,r))+T_{r,r}~\text{over~all}~ r \in R_{p,q}   &~~\textbf{\small (D)} \\
             \alpha_{p,q}+ \weight(E(p,q,r))+U_{p,r}~\text{over~all}~ r \in R_{p,q}  & \textbf{\small (E)} \\
         \end{array}
     \right.
\end{equation}
where for every pair $p,r\in P$ such that $p\prec' r$
\begin{equation}\label{eq3weight}
    U_{p,r} ~=~ \max\{T_{r,s}\}~\text{over~all}~ s\in R'_{p\setminus r}.
\end{equation}

{\bf The third step:} For given $p,s\in P$ such that $p_x<s_x$, let $S_{p,s}$\label{page53} be the weight of $\mathcal{S}_{p,s}$. Using similar reasoning as in the previous subsection, one can show that, if $S_0$ is a subset of $P$ such that $RCH(S_0)$ is vertically separable of maximum weight, the following equations calculate $|RCH(S_0)|$:
\begin{equation}\label{eqopt1weight}
    |RCH(S_0)| ~=~ \max\{S_{p,s}\}~\text{over~all}~ p,s\in P ~\text{such that}~ p_x < s_x
\end{equation}
where for each pair of points $p,s\in P$ such that $p_x<s_x$
\begin{equation}\label{eqopt2weight}
     S_{p,s} ~=~ \left\{
         \begin{array}{lr}
             \max_{\{\text{$4$-separators $(p,q,r,s)$\}}}\{T'_{p,q}+T_{r,s}\}~ & p \prec' s\\
             \max_{\{\text{$4$-separators $(p,q,r,s)$\}}}\{U(p,q,s)+V(p,r,s))\}& p \prec s.\\
         \end{array}
         \right.
\end{equation}
Now $U(p,q,s)$ is defined as $T'_{p,q}+\weight(B(v(q),s))$, and $V(p,r,s)$ is defined as $T_{r,s}+\weight(B(p,u(r)))-\weight(B(p,s))$, with $u=u(r)=(r_x,s_y)$ and $v=v(q)=(q_x,p_y)$.

The proof of the next theorem is a straightforward adaptation of the previous arguments.

\begin{theorem}\label{the:MaxWeightRCH}
The \prob{MaxWeightRCH} problem can be solved in $O(n^3)$ time and $O(n^2)$ space.
\end{theorem}

\section{Maximum area orthoconvex polygon}\label{sec:ortho}

Let $\mathcal{R}$\label{page54} be an axis-aligned rectangle in the plane (usually called domain) and let $P\subset\mathcal{R}$ be a set of $n$ points in general position. In the \prob{MaxOrthoconvexPolygon} problem, we look for an orthoconvex polygon $\mathcal{OC}\subset\mathcal{R}$ of maximum area containing no element of $P$ in its interior. Recall that a polygon is orthoconvex if its sides are axis-parallel and its intersection with any horizontal or vertical line is empty, or a line segment. $\mathcal{OC}$ is bounded by four staircases, determined by some points of $P$ (see Figure~\ref{fig:ortho}). Observe the differences between an optimal solution for the \prob{MaxAreaRCH} problem and an optimal solution for the \prob{MaxOrthoconvexPolygon} problem (see Figures~\ref{fig:separators1} and~\ref{fig:ortho}). In both cases, the boundary of the solutions is defined by four staircases, but in the second case the $1$- and the $3$-staircases (and the $2$- and the $4$-staircases) interchange their roles. Thus, the techniques previously explained must be adapted to this new situation.

\begin{figure}[ht]
	\centering
	\includegraphics[scale=0.85,page=17]{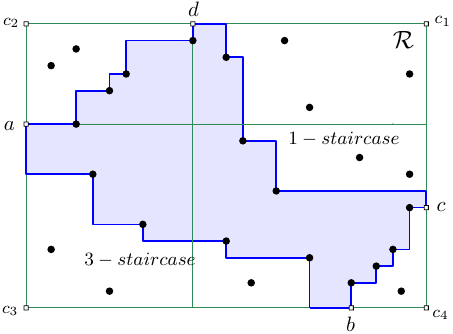}
	\caption{\small{The maximum area orthoconvex polygon.}}\label{fig:ortho}
\end{figure}

As in the previous sections, our algorithm to solve \prob{MaxOrthoconvexPolygon} problem is divided into three steps. In the first step, the algorithm calculates empty staircases. In the second step, the algorithm computes empty orthoconvex polygons of maximum area bounded by three staircases. Finally, in the third step, we combine some of these empty orthoconvex polygons to find an optimal solution. The two main differences in relation to the previous algorithms are that we use $2$-separators (yet to be defined) instead of $4$-separators in the third step, and that the empty orthoconvex polygons used to find the optimal solution must also be defined for the orthogonal projections of the points of $P$ onto the sides of $\mathcal{R}$. This results in recurrences that are a bit more elaborate in the two first steps. We give some definitions, before describing the algorithm.

\begin{definition}
We use $c_1$, $c_2$, $c_3$, and $c_4$\label{page55} to denote the top-right, top-left, bottom-left, and bottom-right vertices of~$\mathcal{R}$, respectively. Let $\bar{P}_t$ (resp., $\bar{P}_{\ell}$, $\bar{P}_b$, $\bar{P}_r$)\label{page56} be the orthogonal projections of the points in $P$ onto the top (resp., left, bottom, right) side of $\mathcal{R}$, and let $\bar{P}=\bar{P}_t\cup \bar{P}_{\ell}\cup \bar{P}_b\cup \bar{P}_r \cup\{c_1,c_2,c_3,c_4\}$\label{page57}.
\end{definition}

Observe that for an optimal solution $\mathcal{OC}$, the boundary of $\mathcal{OC}$ always shares four segments with $\mathcal{R}$, each of them on one different side of $\mathcal{R}$ (see Figure~\ref{fig:ortho}). Let $d$ (resp., $a,b,c$) be the leftmost (resp., topmost, rightmost, bottommost) vertex of $\mathcal{OC}$ on the top (resp., left, bottom, right) side of $\mathcal{OC}$, as shown in Figure~\ref{fig:ortho}. Since these four points belong to $\bar{P}$, we have to define structures whose endpoints belong to $\bar{P}$ in the different steps of our algorithm, as we explain later.

\begin{definition}
For a point $p\in P$, let $x(p)$ and $y(p)$ be the orthogonal projections of $p$ onto the bottom and right sides of $\mathcal R$, respectively. Note that $x(p)\in \bar{P}_b$ and $y(p)\in \bar{P}_r$.
\end{definition}

For an optimal solution $\mathcal{OC}$, Lemma~\ref{lem:canonical} also holds. Therefore, we can assume again that rectangles $B(a,d)$ and $B(b,c)$ are separated by a vertical line, so $\mathcal{OC}$ is vertically separable. We now describe how to find an optimal solution of the \prob{MaxOrthoconvexPolygon} problem that is vertically separable.

{\bf The first step:} In this step, we build a table $D$ whose entries $D_{p,q}$\label{page58} contain areas that are associated with some empty $4$-staircases. Let $p$ and $q$ be a pair of points such that $p\in P\cup \bar{P}_b$ and $q\in P\cup \bar{P}_t\cup\{c_2\}$. If $p\in P$ and $q\not\prec p$, then $D_{p,q}$ is the area of the set $B(x(p),y(q))\setminus M_4(P(x(p),y(q)))$. If $p\in \bar{P}_b$, $p'$ is the point of $P$ such that $x(p')=p$ and $q\not\prec p'$, then $D_{p,q}$ is the area of the set $B(p,y(q))\setminus M_4(P(p,y(q)))\cup \{p'\}$. Figure~\ref{fig:D} shows some examples of $D_{p,q}$.

\begin{figure}[ht]
	\centering
	\includegraphics[scale=0.85,page=37]{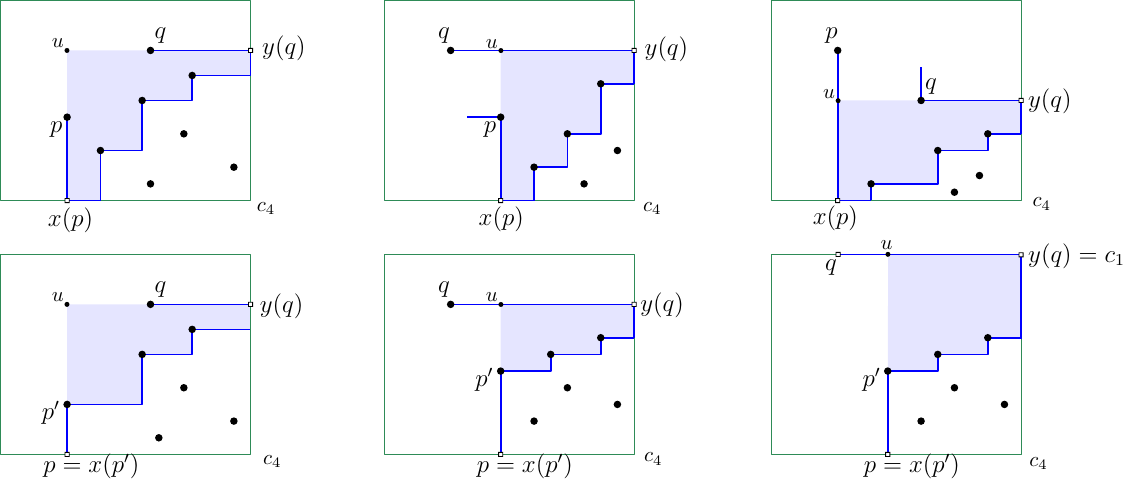}
	\caption{\small{Examples of $D_{p,q}$.}}\label{fig:D}
\end{figure}

For every pair $p,q$, $M_4(P(x(p),y(q))$ (or $M_4(P(p,y(q)))\cup \{p'\}$) can easily be calculated in $O(n)$ time, so also $D_{p,q}$. Therefore, table $D$ can be filled in $O(n^3)$ time\footnote{This approach of defining the area as $B(x(p),y(q))\setminus M_4(P(x(p),y(q)))$ can also be used in the first step of Sections~\ref{sec:EmptyRectCH_1} and~\ref{sec:EmptyRectCH_2}, as an alternative way to compute table $C$.}.

{\bf The second step:} Let $p\in P\cup \bar{P}_b\cup \{c_4\}$, $q\in P\cup \bar{P}_t \cup \{c_2\}$, and $u=(p_x,q_y)$. We define the table $O$ with entries $O_{p,q}$ equal to the area of the orthoconvex empty polygon $\mathcal O_{p,q}$\label{page59} of maximum area, with some restrictions on $p$, $q$, and $u$. Let us suppose first that $p\in P$. The value $O_{p,q}$\label{page60} is the area of a maximum orthoconvex empty polygon with input points $P\cap Q_4(u)$ and domain $\mathcal{R}\cap \overline{Q_4(u)}$ such that the polygon contains $p$, $q$, and $u$ if $p\prec q$ (see Figure~\ref{fig:O} top-left), contains $p$ and $u$ if $q\prec 'p$ (see Figure~\ref{fig:O} top-middle) and contains $q$ and $u$ if $p\prec 'q$ (see Figure~\ref{fig:O} top-right). As we will see in Lemma~\ref{lemm:eq13}, we can calculate each entry $O_{p,q}$ with $p\in P$, using the following recurrence:

\begin{figure}[ht]
	\centering
	\includegraphics[scale=0.8,page=38]{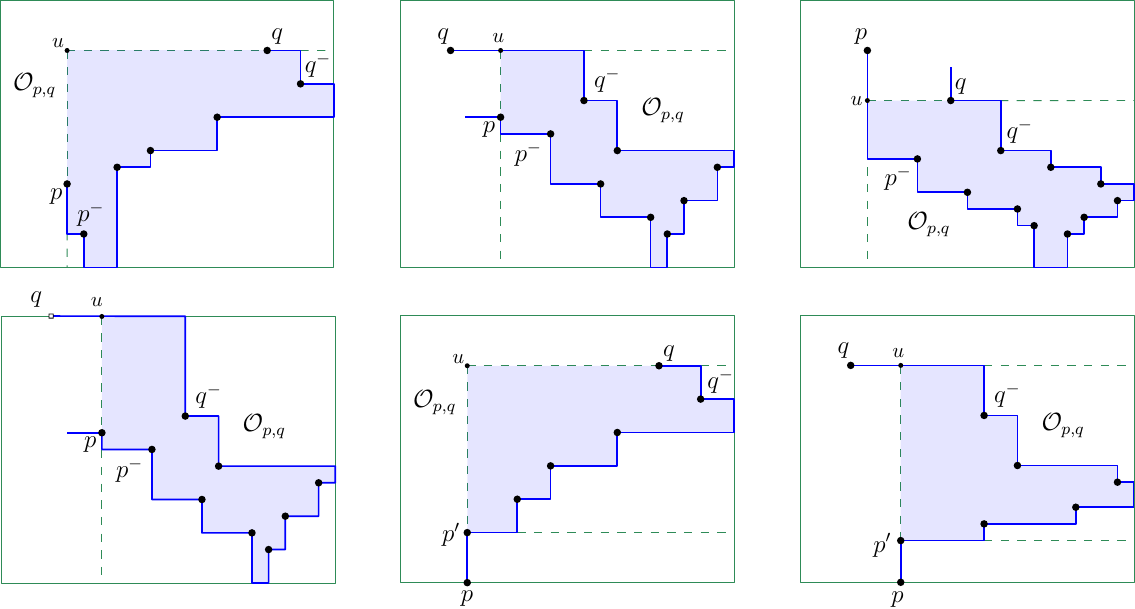}
	\caption{\small{Examples of $\mathcal O_{p,q}$.}}\label{fig:O}
\end{figure}

\begin{equation}\label{eq11}
    \arraycolsep=1.4pt\def\arraystretch{1.2}
    O_{p,q} ~=~ \max\left\{
        \begin{array}{lr}
            D_{p,q}, & \textbf{\small (A)} \\
            \area(B(u,r))+O_{r,q}~\text{over~all}~
                    r\in N_{p,q} & ~\textbf{\small (B)} \\
            \area(B(u,r))+O_{p,r}~\text{over~all}~
                    r\in N_{q,p} &  \textbf{\small (C)}
        \end{array}
    \right.
\end{equation}
where
\[
    N_{p,q} ~=~ \left\{
        \begin{array}{ll}
            r\in P\cap Q_4(p):P(u,r)=\emptyset & \mathbf{if~} p\prec q \mathbf{~or~} q\prec' p \\
            r\in P\cap Q_4(u):P(u,r)=\emptyset  & \mathbf{if~} p\prec' q,
        \end{array}
    \right.
\]
and
\[
    N_{q,p} ~=~ \left\{
        \begin{array}{ll}
            r\in P\cap Q_4(q):P(u,r)=\emptyset & \mathbf{if~} p\prec q \mathbf{~or~} p\prec' q \\
            r\in P\cap Q_4(u):P(u,r)=\emptyset  & \mathbf{if~} q\prec' p.
        \end{array}
    \right.
\]
Suppose  that $p\in \bar{P}_b$ and let $p'$ be the point of $P$ such that $x(p')=p$. In this case, $O_{p,q}$ is the area of a maximum orthoconvex empty polygon with input points $P\cap Q_4(u)\cap Q_1(p')$ and domain $\mathcal{R}\cap \overline{Q_4(u)} \cap\overline{Q_1(p')}$, such that the polygon contains $p'$, $q$, and $u$ if $p'\prec q$ (see Figure~\ref{fig:O} bottom-middle) and contains $p'$ and $u$ if $q\prec'p'$ (see Figure~\ref{fig:O} bottom-right). We calculate each entry $O_{p,q}$ with $p\in \bar{P}_b$, using the following recurrence:
\begin{equation}\label{eq11b}
	\arraycolsep=1.4pt\def\arraystretch{1.2}
	O_{p,q} ~=~ \max\left\{
		\begin{array}{lr}
			D_{p,q}, & \textbf{\small (A)} \\
			\area(B(u,r))+O_{p,r}~\text{over~all}~
				r\in N'_{q,p} &  ~~\textbf{\small (C)}
		\end{array}
	\right.
\end{equation}
where
\[
    N'_{q,p} ~=~ \left\{
        \begin{array}{ll}
            r\in P\cap Q_4(q)\cap Q_1(p') :P(u,r)=\emptyset & \mathbf{if~} p'\prec q \\
            r\in P\cap Q_4(u)\cap Q_1(p') :P(u,r)=\emptyset & \mathbf{if~} q\prec ' p'.
        \end{array}
    \right.
\]
Finally, suppose that $p=c_4$. For any point $q$, an orthoconvex polygon in region $\mathcal{R}\cap \overline{Q_4(u)}$ that contains  $p$ and $u$ degenerates to the segment $up$. Thus, we define $O_{p,q}=0$ in this case. The following lemma proves the correctness of Equations~(\ref{eq11}) and~(\ref{eq11b}).

\begin{figure}[ht]
	\centering
	\includegraphics[scale=0.9,page=22]{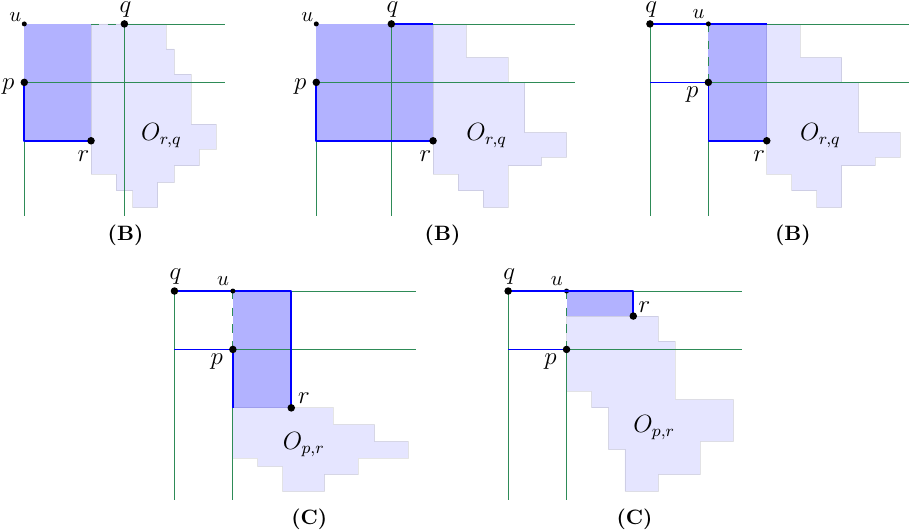}
	\caption{\small{Some cases in the recursive computation of $O_{p,q}$, $p\in P$.}}\label{fig:cases-ortho}
\end{figure}

\begin{lemma}\label{lemm:eq13}
The previous recurrences correctly calculate all the $O_{p,q}$ values in $O(n^3)$ time and $O(n^2)$ space.
\end{lemma}

\begin{proof}
Assume first that $p\in P$. Given $\mathcal O_{p,q}$, let $p^-\in P\cup \bar{P}_b$ be the vertex of $\mathcal O_{p,q}$ that follows $p$ (or $u$ when $p\notin \mathcal O_{p,q}$) in the anti-clockwise traversal of the boundary of $\mathcal O_{p,q}$ (see Figure~\ref{fig:O}). Similarly, let $q^-\in P\cup \bar{P}_r \cup \{c_1\}$ be the vertex of $\mathcal O_{p,q}$ that precedes $q$ (or $u$ when $q\notin \mathcal O_{p,q}$) in the counter-clockwise traversal of the boundary of $\mathcal O_{p,q}$ (see Figure~\ref{fig:O}). Note that we are excluding the points of $\bar{P}_t$ in the definition of $q^-$, although some of them could be a vertex of $\mathcal O_{p,q}$ when $q\in \bar{P}_t$ (see Figure~\ref{fig:O} bottom-left). Also note that $q^-$ can be $c_1$ only if $q\in \bar{P}_t$.

If $p^-\in \bar{P}_b$ and $q^-\in \bar{P}_r\cup \{c_1\}$, then $O_{p,q}=D_{p,q}$, and this fits into case $\textbf{\small (A)}$ of Equation~\eqref{eq11}. Otherwise, $p^-$ or $q^-$ (or both) belongs to $P$. Note that if $p^-\in P$ then necessarily $p^-\in N_{p,q}$, as $P(u,r)=\emptyset$ and $p^-$ is in either $Q_4(p)$ (if $p\prec q$ or $q\prec' p$) or $Q_4(u)$ (if $p\prec' q$). Similarly, if $q^-\in P$ then $q^-\in N_{q,p}$. In the first case we have $O_{p,q}=\area(B(u,p^-))+O_{p^-,q}$, and letting $r=p^-$ this fits into case $\textbf{\small (B)}$ of Equation~\eqref{eq11} (see Figure~\ref{fig:cases-ortho}). In the second case we have $O_{p,q}=\area(B(u,q^-))+O_{p,q^-}$, and letting $r=q^-$ this fits into case $\textbf{\small (C)}$ of Equation~\eqref{eq11}.

The case in which $p\in \bar{P}_b$ is similar to prove, only using point $q^-$. If $q^-\in \bar{P}_r\cup \{c_1\}$, then  $O_{p,q}=D_{p,q}$. If $q^-\in P$, then $q^-\in N'_{q,p}$ and $O_{p,q}=\area(B(u,q^-))+O_{p,q^-}$. Letting $r=q^-$ this fits into case $\textbf{\small (C)}$ of Equation~\eqref{eq11b}. Hence, Equations~\eqref{eq11} and~\eqref{eq11b} correctly calculate all the $O_{p,q}$'s.
\end{proof}

To finish this step, let $\mathcal{O}'_{p,q}$ denote the orthoconvex polygon $\mathcal{O}_{q,p}$ after rotating the coordinates by $\pi$. Let $O'_{p,q}$ denote the area of $\mathcal{O}'_{p,q}$. Note that by symmetry, all the $O'_{p,q}$ values can also be calculated in $O(n^3)$ time.

{\bf The third step:} We show now that by combining some orthoconvex polygons $\mathcal{O}_{p,q}$ and $\mathcal{O}'_{p,q}$, we can find a vertically separable orthoconvex polygon $\mathcal{OC}\subset\mathcal{R}$ of maximum area, whose interior does not contain elements of $P$. To this end, we next define a $2$-separator. Let $p\in P\cup \bar{P}_b\cup\{c_4\}$ and $q\in P\cup \bar{P}_t\cup\{c_2\}$.

\begin{definition}
We say that $(p,q)$ is a \emph{$2$-separator} of $\mathcal{OC}$ if: $(1)$ $p$ belongs to the $3$-staircase of $\mathcal{OC}$; $(2)$ $q$ belongs to the $1$-staircase of $\mathcal{OC}$; $(3)$ $q\prec'p$; and $(4)$ $\mathcal{OC}$ is the union of the pairwise interior-disjoint polygons, $\mathcal{O}'_{p,q}$, $\overline{B(p,q)}$, and $\mathcal{O}_{p,q}$, where $\overline{B(p,q)}$ is the closure of the open rectangle $B(p,q)$ (see Figure~\ref{fig:4step-orth0}).
\end{definition}

\begin{figure}[ht]
	\centering
	\includegraphics[scale=0.9,page=33]{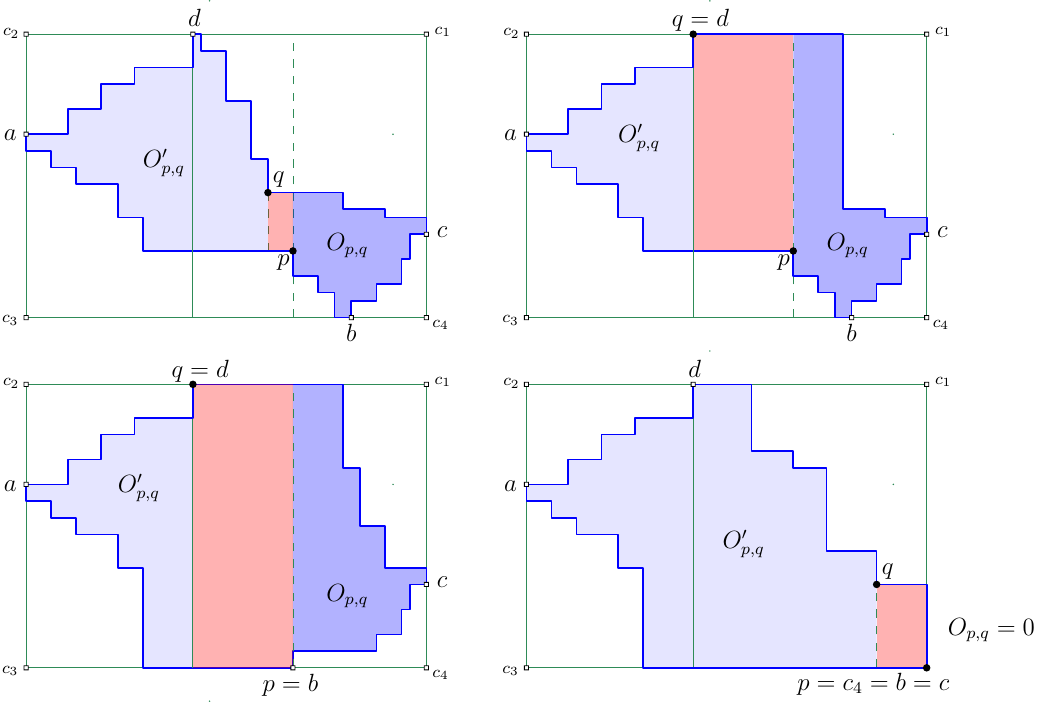}
	\caption{\small{Third step of the algorithm: $2$-separators $(p,q)$.}}\label{fig:4step-orth0}
\end{figure}

This definition implies that if the optimal solution $\mathcal{OC}$ has a $2$-separator $(p,q)$, then $$\area(\mathcal{OC})=O'_{p,q}+\area(B(p,q))+O_{p,q}.$$
The next lemma establishes the equivalence between vertical separability and $2$-separators.

\begin{lemma}\label{lem:separator2}
$\mathcal{OC}$ is vertically separable if and only if $\mathcal{OC}$ has a $2$-separator.
\end{lemma}

\begin{proof}
Let $\mathcal{OC}\subset\mathcal{R}$ be an orthoconvex polygon of maximum area that is vertically separable, i.e., rectangles $B(a,d)$ and $B(b,c)$ are separated by a vertical line. If there exist elements of $P$ to the right of the vertical line $\ell_d$ through $d$ that are vertices of the $3$-staircase of $\mathcal{OC}$, then let $p$ be the leftmost of them. Otherwise, let $p=b$. If there exist elements of $P$ between $\ell_d$ and the vertical line through $p$ that are vertices of the $1$-staircase of $\mathcal{OC}$, then let $q$ be the rightmost of them. Otherwise, let $q=d$. For the definitions of $p$ and $q$, refer to Figure~\ref{fig:4step-orth0}. It is easy to see that $(p,q)$ is a $2$-separator of $\mathcal{OC}$.

Let $(p,q)$ be a $2$ separator of $\mathcal{OC}$. Since $p$ belongs to the $3$-staircase, we have $p=b$ or $p\prec' b$. Similarly, since $q$ belongs to the $1$-staircase, we have $q=d$ or $d\prec' q$. Given that $q\prec' p$, we have by transitivity that $d\prec' b$, which implies $d_x<b_x$. Therefore, rectangles $B(a,d)$ and $B(b,c)$ are separated by a vertical line, so $\mathcal{OC}$ is vertically separable.
\end{proof}

Following Lemma~\ref{lem:separator2}, the algorithm to find $\mathcal{OC}$ is simple: For every $p\in P\cup \bar{P}_b\cup\{c_4\}$ and $q\in P\cup \bar{P}_t\cup\{c_2\}$, we compute the orthoconvex polygon $\mathcal{OC}_{p,q}$ of maximum area (equals to $O'_{p,q}+\area(B(p,q))+O_{p,q}$), such that $(p,q)$ is a $2$-separator of $\mathcal{OC}_{p,q}$. The optimal solution $\mathcal{OC}$ will be one of these $\mathcal{OC}_{p,q}$ values. Summarizing over all $2$-separators $(p,q)$ with $p\in P\cup \bar{P}_b\cup\{c_4\}$ and $q\in P\cup \bar{P}_t\cup\{c_2\}$:
\begin{equation}\label{eqoc1}
    \area(\mathcal{OC}) ~=~ \max\{O'_{p,q}+\area(B(p,q))+O_{p,q}\}.
\end{equation}

This last step requires $O(n^2)$ time, once all the $O_{p,q}$ values and all the $O'_{p,q}$ values have been calculated. Thus, we have obtained a new $O(n^3)$ time and $O(n^2)$ space algorithm for the \prob{MaxOrthoconvexPolygon} problem which is simpler than the one obtained by Nandy et al.~\cite{nandy2010}.

\begin{theorem}[Nandy et al.~\cite{nandy2010}]
The \prob{MaxOrthoconvexPolygon} problem can be solved in $O(n^3)$ time and $O(n^2)$ space.
\end{theorem}

\section{Empty staircase polygon with the largest area}\label{sec:MESP}

In this section, we further extend the applications of our techniques above to show that finding an empty staircase polygon with the largest area amidst the $n$-point set $P\subset\mathcal{R}$ (i.e., the \prob{MaxStaircasePolygon} problem) can be done in $O(n^2)$ time and space. Recall that given a domain $\mathcal{R}$ and a point set $P\subset\mathcal{R}$, a maximum-area empty staircase polygon is an orthoconvex polygon contained in $\mathcal{R}$ with no point of $P$ in the interior, that includes two opposed corners of $\mathcal{R}$ as vertices. See Figure~\ref{fig:mesp}, left.

\begin{figure}[ht]
	\centering
	\includegraphics[scale=0.9,page=23]{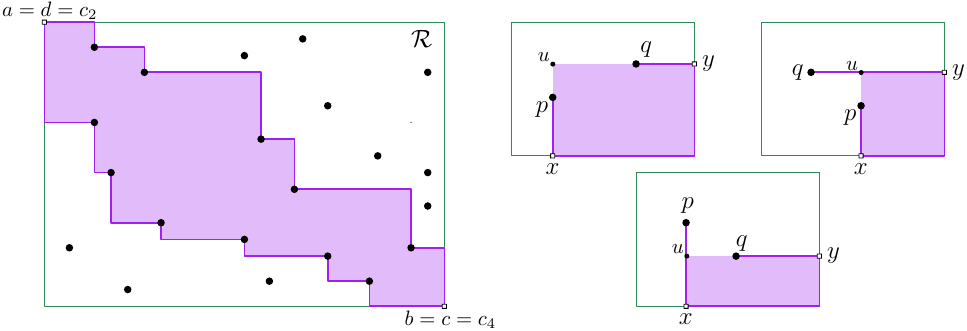}
	\caption{\small{A maximum area empty staircase polygon, and the three cases for $D_{p,q}$.}}\label{fig:mesp}
\end{figure}

The algorithm in this section only consists of two steps: In the first step, some empty rectangles are defined. In the second step, we build some empty staircase polygons of maximum area for some pairs of points. One of these polygons will be the optimal solution for the \prob{MaxStaircasePolygon} problem.

{\bf The first step}: Let $\hat{P}=P\cup \{c_2\}$. For every $p,q\in \hat{P}$ such that $q\not\prec p$, we redefine $D_{p,q}$ to denote the $\area(B(x(p),y(q)))$ when the open rectangle $B(x(p),y(q))$ does not contain elements of $P$ (see Figure~\ref{fig:mesp}, right) and $D_{p,q}=-\infty$ otherwise. Recall that $u=(p_x,q_y)$. Let $D$ be the table containing all of the values $D_{p,q}$. Note that we do not need to explicitly compute $D$, since each entry can be computed in $O(1)$ time on demand.

{\bf The second step}: Consider a new table $O$, whose entry $O_{p,q}$ is the area of a maximum-area empty staircase polygon $\mathcal{O}_{p,q}$\label{page61} defined for: (1) $p=q=c_2$; (2) $p,q\in P$, $p\prec q$, $P(p,q)=\emptyset$; and (3) $p,q\in\hat{P}$, $q\prec' p$ or $p\prec' q$. Recall that these empty staircase polygons $\mathcal{O}_{p,q}$ are built with input points $P\cap Q_4(u)$ and domain $\mathcal{R}\cap\overline{Q_4(u)}$ such that the polygon contains $p$, $q$, and $u$ if $p\prec q$, contains $p$ and $u$ if $q\prec' p$, and contains $q$ and $u$ if $p\prec' q$. See Figure~\ref{fig:stairp} for some examples.

\begin{figure}[ht]
	\centering
	\includegraphics[scale=0.8,page=28]{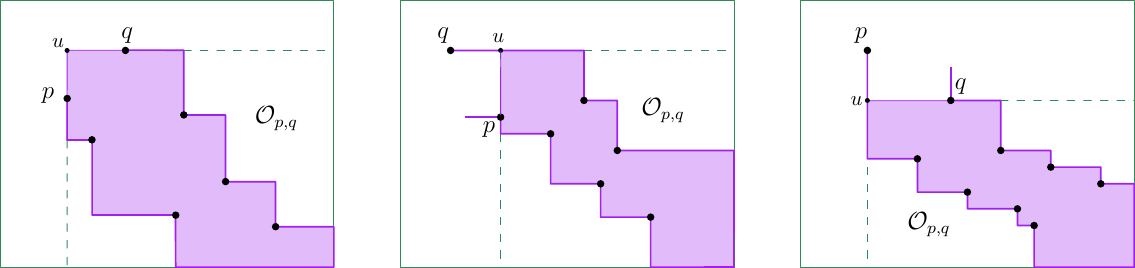}
	\caption{\small{Some examples of $\mathcal{O}_{p,q}$.}}\label{fig:stairp}
\end{figure}

When $p\prec q$, note that $\mathcal{O}_{p,q}$ always contains $B(p,q)$ in its interior, hence necessarily $P(p,q)=\emptyset$ so that $\mathcal{O}_{p,q}$ is defined. This is the reason why condition $P(p,q)=\emptyset$ was added to item (2) in the definition. This condition has more implications when calculating all the values $\mathcal{O}_{p,q}$. Suppose that $p\prec q$ and take a point $r\in R_{p\setminus q}$ (see Figure~\ref{fig:next}, where $r_p$ in the figure plays the role of $r$). If $\mathcal{O}_{p,q}$ is obtained from $\mathcal{O}_{r,q}$ and $B(u,r)$, then necessarily $P(u,r)=\emptyset$, implying that $P(p,r)=\emptyset$. On the other hand, in $\mathcal{O}_{r,q}$ the set $P(q,r)$ must be necessarily empty. By adding these two restrictions $P(u,r)=P(q,r)=\emptyset$ when calculating $\mathcal{O}_{p,q}$, we reduce the number of points $r$ to consider, as there exists only one point $r$ in $R_{p\setminus q}$ such that $P(p,r)$ and $P(q,r)$ are empty at the same time.

Using a detailed case analysis, the following recurrence including restrictions on the emptiness of several subsets, allows us to calculate all the values $O_{p,q}$ in $O(n^2)$ time:
\begin{equation}\label{eq13}
    \arraycolsep=1.4pt\def\arraystretch{1.2}
    O_{p,q} ~=~ \max\left\{
        \begin{array}{lr}
            D_{p,q} & \textbf{\small (A)} \\
            \area(B(u,r))+O_{r,q}~\text{over~all}~r\in M_{p,q} & ~\textbf{\small (B)}\\
            \area(B(u,r))+O_{p,r}~\text{over~all}~r\in M_{q,p} &  \textbf{\small (C)}
        \end{array}
    \right.
\end{equation}
where
\[
	M_{p,q} ~=~ \left\{
		\begin{array}{ll}
			r\in P\cap Q_4(p):P(u,r)=P(q,r)=\emptyset & \mathbf{if~} p\prec q\\
			r\in P\cap Q_4(p):P(u,r)=\emptyset        & \mathbf{if~} q\prec' p\\
			r\in P\cap Q_4(u):P(u,r)=P(q,r)=\emptyset & \mathbf{if~} p\prec' q
		\end{array}
	\right.
\]
and
\[
	M_{q,p} ~=~ \left\{
		\begin{array}{ll}
			r\in P\cap Q_4(q):P(u,r)=P(p,r)=\emptyset & \mathbf{if~} p\prec q\\
			r\in P\cap Q_4(u):P(u,r)=P(p,r)=\emptyset & \mathbf{if~} q\prec' p\\
			r\in P\cap Q_4(q):P(u,r)=\emptyset        & \mathbf{if~} p\prec' q.
		\end{array}
	\right.
\]
We remark that item (2) $p,q\in P$, $p\prec q$, $P(p,q)=\emptyset$ could be replaced by (2) $p,q\in P$, $p\prec q$, as the previous recurrence assigns the value $-\infty$ to $O_{p,q}$ (meaning that $\mathcal O_{p,q}$ does not exist), when $p\prec q$ and $P(p,q)\neq \emptyset$. In the definition of $M_{p,q}$ and $M_{q,p}$ when $p\prec q$, observe that for a point $r\in M_{p,q}\cup M_{q,p}$ we have $P(u,r)=P(q,r)=P(p,r)=P(p,q)=\emptyset$. In particular, this implies that if $P(p,q)\neq \emptyset$ then $M_{p,q}\cup M_{q,p}=\emptyset$. Thus, when applying Equation~\eqref{eq13} to the pair $p\prec q$ with $P(p,q)\neq \emptyset$, only case (A) applies, as $M_{p,q}\cup M_{q,p}=\emptyset$. Since $B(x(p),y(q))$ contains points of $P$, then $O_{q,p}=D_{q,p}=-\infty$ by definition.

The correctness of Equation~\eqref{eq13} follows in a similar way as Equation~\eqref{eq11}: If a point $p^-\in P$ follows $p$ (or $u$) in the $3$-staircase of $\mathcal O_{p,q}$, then $\mathcal O_{p,q}$ is obtained from $\mathcal O_{p^-,q}$ and $B(u,p^-)$ (Case (B)). In the same way, if a point $q^-\in P$ precedes $q$ (or $u$) in the $1$-staircase of $\mathcal O_{p,q}$, then $\mathcal O_{p,q}$ is obtained from $\mathcal O_{p,q^-}$ and $B(u,q^-)$ (Case (C)). If such points do not belong to $P$, then $\mathcal O_{p,q}$ is $B(x(p),y(p))$ (Case (A)).

\begin{figure}[ht]
	\centering
	\subfloat[]{\includegraphics[scale=1.05,page=29]{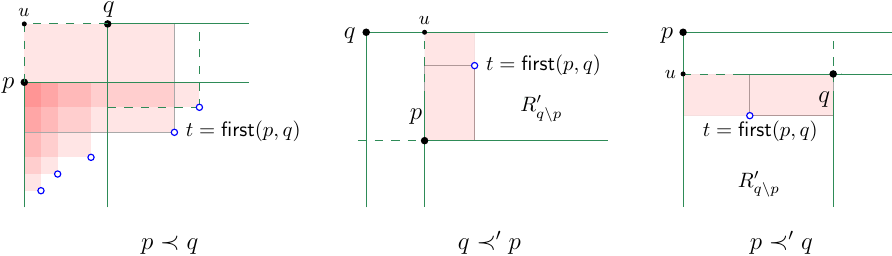}\label{fig:first}}\\
	\subfloat[]{\includegraphics[scale=1.05,page=30]{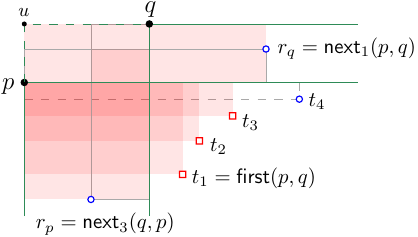}\label{fig:next}}
    \subfloat[]{\includegraphics[scale=1.05,page=39]{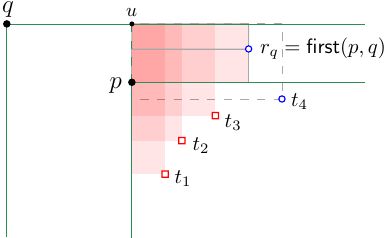}\label{fig:other}}
	\caption{\small{(a) The definition of $\first(p,q)$. (b) The case $p\prec q$ in the proof of Lemma~\ref{lem:iteration}, where $r_p\neq\nil$ and $r_q\neq\nil$: The points of $L_4(p)$ that are in $R_{p,q}$ are $t_1=\first(p,q)$, $t_2$, $t_3$, and $t_4$, and $R_{p,q}\cap M_{p,q}=R_{p,q}\cap M_{q,p}=\{t_1,t_2,t_3\}$. Note that $t_4$
is not present in this set because $r_q\in P(q,t_4)$. (c) The case $q\prec' p$ in the proof of Lemma~\ref{lem:iteration}, where $r_q\neq\nil$: The points of $L_4(p)$ that are in $R_{p,q}\cap M_{p,q}=R_{p,q}\cap M_{q,p}$ are $\{t_1,t_2,t_3\}$.}}\label{fig:first-next}
\end{figure}

Observe that $O_{c_2,c_2}$ is the area of the maximum-area empty staircase polygon amidst $P$, so we only need to show that Equation~\eqref{eq13} calculates all the values $O_{p,q}$ in $O(n^2)$ time, to solve the \prob{MaxStaircasePolygon} problem with this complexity.

For every $p\in\hat{P}$, let $L_i(p)$\label{page62}, $i\in\{1,2,3,4\}$, be the sequence of the points $t\in \hat{P}\cap Q_i(p)$ sorted by increasing $x$-coordinate such that $P(p,t)=\emptyset$. Since $P$ is already sorted by  $x$-coordinate, the sequence $L_i(p)$ can be computed in $O(n)$ time for each point $p$, thus $L_i(p)$ is computed in $O(n^2)$ time and space.

Let $K\subset \hat{P}\times \hat{P}$ be the set of point pairs $(p,q)$ such that the entry $O_{p,q}$ is defined. The next lemma shows that for a pair $(p,q)\in K$, the complexity of calculating the sets $M_{p,q}$ and $M_{q,p}$ is linear in their size, after an $O(n^2)$-time preprocessing.

\begin{lemma}\label{lem:iteration}
For each $(p,q)\in K$, the sets $M_{p,q}$ and $M_{q,p}$ can be computed, respectively, in $O(|M_{p,q}|)$ and $O(|M_{q,p}|)$ time, after an $O(n^2)$-time preprocessing.
\end{lemma}

\begin{proof}
We give an $O(n^2)$-time preprocessing to compute some special points that we call $\first(p,q)$. For every $t\in L_i(p)$, let $\next_i(p,t)$ be the point that goes after $t$ in $L_i(p)$. If $t$ is the last element, then $\next_i(p,t)=\nil$. We define $\first(p,q)$ for the next pairs $p,q\in \hat{P}$ as follows (see Figure~\ref{fig:first}):

\begin{itemize}
\item $p\prec q$ and $P(p,q)=\emptyset$: Let $\first(p,q)$ be the leftmost point $t$ in $L_4(p)$ such that $q\prec't$ and $P(q,t)=\emptyset$. Note that $q$ is in $L_1(p)$, and by simultaneously traversing $L_1(p)$ and $L_4(p)$, we can compute $\first(p,q)$ in $O(n)$ time for fixed $p$ and all $q$ such that $p\prec q$.

\item $q\prec' p$: Let $\first(p,q)$ be the leftmost point $t$ in $R'_{q\setminus p}$, which ensures $P(u,t)=P(p,t)=\emptyset$. Recall that when $q\prec' p$, $R'_{q\setminus p}$ is the subset of $P$ in the region $Q_4(q)\cap Q_1(p)$. Note that $t$ is in $L_1(p)$, and by simultaneously traversing the $y$-ordering of $\hat{P}$ and $L_1(p)$, we can compute $\first(p,q)$ in $O(n)$ time for fixed $p$ and all $q\prec' p$.

\item $p\prec' q$: Let $\first(p,q)$ be the topmost point $t$ in $R'_{q\setminus p}$, which ensures $P(u,t)=P(q,t)=\emptyset$. Recall that when $p\prec' q$, $R'_{q\setminus p}$ is the subset of $P$ in the region $Q_4(p)\cap Q_3(q)$. Note that $t$ is in $L_3(q)$, and by simultaneously traversing the $x$-ordering of $\hat{P}$ and $L_3(q)$, we can compute $\first(p,q)$ in $O(n)$ time for fixed $q$ and all $p\prec' q$.
\end{itemize}

In any of the above three cases, if the point $t$ does not exist, then $\first(p,q)=\nil$. Observe that we can compute $\first(p,q)$ for all $p,q$ such that $\first(p,q)$ is defined in  $O(n^2)$ time overall.

We now show how to compute $M_{p,q}$ and $M_{q,p}$ in $O(|M_{p,q}|)$ and $O(|M_{q,p}|)$ time, respectively. Consider the case where $p\prec q$, in which $P(p,q)=\emptyset$. Recall that for a point $r\in M_{p,q}\cup M_{q,p}$, we have $P(u,r)=P(q,r)=P(p,r)=P(p,q)=\emptyset$. Let $r_p=\next_3(q,p)$ and $r_q=\next_1(p,q)$, and assume without loss of generality that $r_p\neq \nil$ and $r_q\neq\nil$ (see Figure~\ref{fig:next}).

If there exists a point $r\in R_{p\setminus q}\cap M_{p,q}$, then the condition $P(p,r)=P(q,r)=\emptyset$ implies $r=r_p$. Similarly, $r\in R_{q\setminus p}\cap M_{q,p}$ implies $r=r_q$. Now, observe that $P(u,r)=\emptyset$ also implies that $R_{p,q}\cap M_{p,q}=R_{p,q}\cap M_{q,p}$. Moreover, if there are elements in these two equal sets, then they appear in $L_4(p)$ as consecutive elements, from $\first(p,q)$ to the right, until the last element in $L_4(p)$ to the left of the vertical line passing through $r_q$. Then, all of the observations together with the fact that we check $P(u,r)=P(p,r)=P(q,r)=\emptyset$ in $O(1)$ time, allow us to calculate $M_{p,q}$ and $M_{q,p}$ in $O(|M_{p,q}|)$ and $O(|M_{q,p}|)$ time, respectively. Note that $r_p$ and $\first(p,q)$ are consecutive in $L_4(p)$.

Suppose now that $q\prec' p$, see Figure~\ref{fig:other}. Let $r_q=\first(p,q)$ and assume without loss of generality that $r_q\neq \nil$. If there exists a point $r\in R'_{q\setminus p}\cap M_{q,p}$, then the condition $P(u,r)=P(p,r)=\emptyset$ implies $r=r_q$. Furthermore, the condition $P(u,r)=\emptyset$ implies $R'_{p,q}\cap M_{p,q}=R'_{p,q}\cap M_{q,p}$, where the elements of this set form a prefix of $L_4(p)$ until the last element in $L_4(p)$ to the left of the vertical line passing through $r_q$. Hence, we can calculate $M_{p,q}$ and $M_{q,p}$ in $O(|M_{p,q}|)$ and $O(|M_{q,p}|)$ time, respectively. The case $p\prec' q$ is symmetric.
\end{proof}

The following lemma proves that the number of triples $(p,q,r)$ such that $O_{p,q}$ is defined and that $r\in M_{p,q}\cup M_{q,p}$ is $O(n^2)$.

\begin{lemma}\label{lem:edges}
$\sum_{(p,q)\in K} \left( |M_{p,q}|+|M_{q,p}| \right)~\in~O(n^2).$
\end{lemma}

\begin{proof}
$\sum_{(p,q)\in K} \left( |M_{p,q}|+|M_{q,p}| \right)$ counts the number of triples $(p,q,r)$ such that $O_{p,q}$ is defined and $r\in M_{p,q}\cup M_{q,p}$. If $p\prec q$ with $P(p,q)=\emptyset$, then we have: (1) If $r\in R_{p\setminus q}\cap M_{p,q}$, then $p$ and $r$ are consecutive in $L_3(q)$; (2) If $r\in R_{p,q}\cap M_{p,q}= R_{p,q}\cap M_{q,p}$, then $p$ and $q$ are consecutive in $L_2(r)$; (3) If $r\in R_{q\setminus p}\cap M_{q,p}$, then $q$ and $r$ are consecutive in $L_1(p)$. Since there are $n$ choices for a point $s\in P$, and for each $s$ there are at most $O(n)$ choices for two consecutive points in any $L_i(s)$, there are at most $O(n^2)$ triples $(p,q,r)$ with $p\prec q$.

Consider $q\prec' p$, and define $U=\{(s,t)\in \hat{P}\times P\mid s\prec' t\}$. If $r\in R'_{q\setminus p}\cap M_{q,p}$, we charge the triplet $(p,q,r)$ to $(q,p)\in U$. Condition $P(u,r)=P(p,r)=\emptyset$ implies that $r$ is unique for the combination $q,p$. Otherwise, if $r\in R'_{p,q}\cap M_{p,q}=R'_{p,q}\cap M_{q,p}$, then we charge the triplet $(p,q,r)$ to $(q,r)\in U$, where $P(u,r)=\emptyset$ implies that $p$ is unique for the combination $q,r$. Then, each element of $U$ is charged to at most two triples, and this ensures that there are at most $O(|U|)=O(n^2)$ triples for $q\prec' p$. Counting for the case $p\prec' q$ is symmetric. Thus the lemma follows.
\end{proof}

The combination of Equation~\eqref{eq13}, Lemma~\ref{lem:iteration}, and Lemma~\ref{lem:edges} allows us to calculate table $O$ in $O(n^2)$ time. We can also find an empty staircase polygon with the largest area $O_{c_2,c_2}$ in $O(n^2)$ time and space. Thus, we have given a new and simpler algorithm for the \prob{MaxStaircasePolygon} problem.

\begin{theorem}[Nandy and Bhattacharya~\cite{nandy2003}]
The \prob{MaxStaircasePolygon} problem can be solved in $O(n^2)$ time and space.
\end{theorem}

To conclude, we remark that the approach given in this section cannot be applied to the \prob{MaxOrthoconvexPolygon} problem described in Section~\ref{sec:ortho}. For example, when $p\prec q$, necessarily $P(p,q)=\emptyset$ in an empty staircase polygon. However, in an empty orthoconvex polygon this property is not true anymore, as the $4$-staircase of the polygon can cross $B(p,q)$ using points belonging to $P(p,q)$. Therefore, the conditions $P(u,r)=P(q,r)=\emptyset$ and $P(u,r)=P(p,r)=\emptyset$  cannot be added to the definition of $N_{p,q}$ and $N_{q,p}$, respectively, because otherwise they would imply that $P(p,q)=\emptyset$.

%
\section*{Acknowledgements}
D.~O. was supported by projects MTM2017-83750-P of the Spanish Ministry of Science (AEI/FEDER, UE) and PID2019-104129GB-I00 / AEI / 10.13039/501100011033 of the Spanish Ministry of Science and Innovation. P.\ P-L.\ was partially supported by projects DICYT 041933PL Vicerrector\'ia de Investigaci\'on, Desarrollo e Innovaci\'on USACH (Chile), and Programa Regional STICAMSUD 19-STIC-02. C.~S. was supported by projects MTM2015-63791-R MINECO/FEDER, Gen.~Cat. DGR 2017SGR1640, and PID2019-104129GB-I00 / AEI / 10.13039/501100011033 of the Spanish Ministry of Science and Innovation. J.~T. was supported by projects MTM2015-63791-R MINECO/FEDER, Gobierno de Arag\'on E41-17R, and PID2019-104129GB-I00 / AEI / 10.13039/501100011033 of the Spanish Ministry of Science and Innovation. J.~U. was supported by PAPIIT grant IN102117 from UNAM.


\newpage
\section*{Summary of notation}\label{notation}

The following table summarizes the main notation used. Horizontal lines denote a change of (sub)section implying a change on the meaning of a notation.

\begin{center}
\begin{tabular}{p{0.1\textwidth}p{0.75\textwidth}p{0.05\textwidth}}
Notation & Description & Page \\
\hline
$P$ & Point set in general position in the plane & \pageref{page1} \\
$a,b,c,d$ & Leftmost, bottommost, rightmost, and topmost points of~$P$ & \pageref{page2} \\
$p_x,p_y$ & $x$- and $y$-coordinates of a point~$p$ & \pageref{page3} \\
$p\prec q$ & Denotes that $p_x < q_x$ and $p_y < q_y$ & \pageref{page4} \\
$p\prec' q$ & Denotes that $p_x < q_x$ and $p_y > q_y$ & \pageref{page5} \\
$Q_1(p)$ & Open axis-aligned quadrant $\{q\in\mathbb{R}^2\mid p\prec q\}$ & \pageref{page6} \\
$Q_2(p)$ & Open axis-aligned quadrant $\{q\in\mathbb{R}^2\mid q\prec' p\}$ & \pageref{page7} \\
$Q_3(p)$ & Open axis-aligned quadrant $\{q\in\mathbb{R}^2\mid q\prec p\}$ & \pageref{page8} \\
$Q_4(p)$ & Open axis-aligned quadrant $\{q\in\mathbb{R}^2\mid p\prec' q\}$ & \pageref{page9} \\
$M_i(P)$ & Union $\bigcup_{p\in P} \overline{Q_i(p)}$ & \pageref{page10} \\
$RCH(P)$ & Rectilinear convex hull of~$P$, defined as $\bigcap_{i=1,2,3,4} M_i(P)$ & \pageref{page11} \\
$B(u,v)$ & Smallest open axis-aligned rectangle containing $u$ and~$v$ & \pageref{page12} \\
$P(u,v)$ & Set of points in~$P$ that belong to $B(u,v)$, i.e., $P\cap B(u,v)$ & \pageref{page13} \\
$R_{p\setminus q}$ & For $p\prec q$, subset of $P$ in the region $Q_4(p)\setminus \overline{Q_4(q)}$ & \pageref{page14} \\
$R_{q\setminus p}$ & For $p\prec q$, subset of $P$ in the region $Q_4(q)\setminus \overline{Q_4(p)}$ & \pageref{page15} \\
$R_{p,q}$ & For $p\prec q$, subset of $P$ in the region $Q_4(p)\cap Q_4(q)$ & \pageref{page16} \\
$R'_{p\setminus q}$ & For $q\prec'p$, subset of $P$ in the region $Q_4(q)\cap Q_3(p)$ & \pageref{page17} \\
$R'_{q\setminus p}$ & For $q\prec'p$, subset of $P$ in the region $Q_4(q)\cap Q_1(p)$ & \pageref{page18} \\
$R'_{p,q}$ & For $q\prec'p$, subset of $P$ in the region $Q_4(p)$ & \pageref{page19} \\
\hline
$\mathcal C_{p,q}$ & For $p\prec q$ or $p=q$, a $2$-staircase with endpoints $p$ and $q$ of maximum size & \pageref{page20} \\
$C_{p,q}$ & Number of elements of $P$ in $\mathcal C_{p,q}$ & \pageref{page21} \\
$Z(p,q)$ 
 & Fourth quadrant associated with $u=(p_x,q_y)$ & \pageref{page22} \\
$z(p,q)$ & Points of $P$ in $Z(p,q)$, i.e., $=Z(p,q)\cap P$ & \pageref{page23} \\
$\mathcal T_{p,q}$ & Triple $1$-staircase of maximum cardinality among all $S\cup \{p,q\}$ with $S\subseteq z(p,q)$ & \pageref{page24} \\
$X_{p,q}$ & Set of extreme vertices of $\mathcal T_{p,q}$ & \pageref{page25} \\
$T_{p,q}$ & Cardinality of $X_{p,q}$ & \pageref{page26} \\
$T$ & Table storing values $T_{p,r}$ & \pageref{page27} \\
$U_{p,r}$ & $\max\{T_{r,s}\}~\text{over~all}~ s\in R'_{p\setminus r}$ & \pageref{eq3} \\
$U$ & Table storing values $U_{p,r}$ & \pageref{page29} \\
$\mathcal T'_{p,q}$ & Triple $3$-staircase of maximum cardinality 
& \pageref{page30} \\
$T'_{p,q}$ & Size of $\mathcal T'_{p,q}$ & \pageref{page31} \\
$\mathcal{S}_{p,s}$ & For $p_x<s_x$, rectilinear convex hull of maximum size 
& \pageref{page32} \\
$S_{p,s}$ & Size of $\mathcal{S}_{p,s}$ & \pageref{page33} \\
\hline
\end{tabular}

\begin{tabular}{p{0.1\textwidth}p{0.75\textwidth}p{0.05\textwidth}}
$\mathcal C_{p,q}$ & Empty $2$-staircase of maximum cardinality 
& \pageref{page34} \\
$\mathcal{O}_{\mathcal{T}}$ & Region $\overline{Z(p,q)}\cap M_1(S\cup \{p,q\})\cap M_2(S\cup \{p,q\})\cap M_3(S\cup \{p,q\})$ 
& \pageref{page35} \\
$\mathcal T_{p,q}$ & Empty triple $1$-staircase of maximum size 
& \pageref{page36} \\
$E(p,q,r)$ & Interior of $B(p,q)\cap Q_2(r)$ if $p\ne q$, the empty set if $p=q$ & \pageref{page37} \\
$T_{p,q}$ & Cardinality of the extreme vertices of $\mathcal T_{p,q}$ & \pageref{page38} \\
$\mathcal T'_{p,q}$ & Empty triple $3$-staircase
& \pageref{page39} \\
$\mathcal{S}_{p,s}$ & For $p_x<s_x$, {empty} rectilinear convex hull of maximum size 
& \pageref{page40} \\
\end{tabular}

\begin{tabular}{p{0.13\textwidth}p{0.72\textwidth}p{0.05\textwidth}}
\hline
$\mathcal C_{p,q}$ & As in the previous appearance, but maximizing area instead of size & \pageref{page41} \\
$\mathcal T_{p,q}$ & As in the previous appearance, but maximizing area instead of size & \pageref{page42} \\
$\mathcal T'_{p,q}$ & As in the previous appearance, but maximizing area instead of size & \pageref{page43} \\
$\mathcal S_{p,q}$ & As in the previous appearance, but maximizing area instead of size & \pageref{page44} \\
$T_{p,q}$ & Area of $\mathcal T_{p,q}$ & \pageref{page45} \\
$S_{p,s}$ & Area of $\mathcal{S}_{p,s}$ & \pageref{page46} \\
\hline
$\mathcal C_{p,q}$ & As in the first appearance, but maximizing weight & \pageref{page47} \\
$\mathcal T_{p,q}$ & As in the first appearance, but maximizing weight & \pageref{page48} \\
$\mathcal T'_{p,q}$ & As in the first appearance, but maximizing weight & \pageref{page49} \\
$\mathcal S_{p,q}$ & As in the first appearance, but maximizing weight & \pageref{page50} \\
$C_{p,q}$ & Weight of $\mathcal C_{p,q}$ & \pageref{page51} \\
$T_{p,q}$ & Weight of $\mathcal T_{p,q}$ & \pageref{page52} \\
$S_{p,s}$ & Weight of $\mathcal{S}_{p,s}$ & \pageref{page53} \\
\hline
$\mathcal{R}$ & Axis-aligned rectangle in the plane & \pageref{page54} \\
$c_1$, $c_2$, $c_3$, $c_4$ & Top-right, top-left, bottom-left, and bottom-right vertices of~$\mathcal{R}$, respectively & \pageref{page55} \\
$\bar{P}_t$, $\bar{P}_{\ell}$, $\bar{P}_b$, $\bar{P}_r$ & Orthogonal projections of $P$ onto the (resp.) top, left, bottom, right side of $\mathcal{R}$ & \pageref{page56} \\
$\bar{P}$ & Union $\bar{P}_t\cup \bar{P}_{\ell}\cup \bar{P}_b\cup \bar{P}_r \cup\{c_1,c_2,c_3,c_4\}$ & \pageref{page57} \\
$D_{p,q}$ & Areas associated with some empty $4$-staircases & \pageref{page58} \\
$\mathcal O_{p,q}$ & Orthoconvex empty polygon of maximum area & \pageref{page59} \\
$O_{p,q}$ & Area of $\mathcal O_{p,q}$ & \pageref{page60} \\
\hline
$\mathcal O_{p,q}$ & Maximum-area empty staircase polygon & \pageref{page61} \\
$L_i(p)$ & Sequence of the points $t\in \hat{P}\cap Q_i(p)$ 
such that $P(p,t)=\emptyset$ & \pageref{page62} \\
\end{tabular}
\end{center}

\end{document}